\documentclass[twocolumn,letterpaper,10pt,draftcls]{IEEEtran}
\usepackage{amsmath,amsthm,amssymb}
\usepackage[ruled,vlined,linesnumbered]{algorithm2e}
\usepackage{array}
\usepackage{textcomp}
\usepackage{stfloats}
\usepackage{url}
\usepackage{verbatim}
\usepackage{graphicx}
\usepackage{cite}
\usepackage{url}
\usepackage{xcolor}
\usepackage{multirow}
\usepackage{subcaption}

\newtheorem{proposition}{Proposition}
\newtheorem{lemma}{Lemma}
\newtheorem{definition}{Definition}
\theoremstyle{remark}
\newtheorem{remark}{Remark}

\newcommand{\R}{\mathbb{R}}
\newcommand{\Z}{\mathbb{Z}}
\newcommand{\mB}{\mathcal{B}}
\newcommand{\mN}{\mathcal{N}}
\newcommand{\mQ}{\mathcal{Q}}
\newcommand{\mG}{\mathcal{G}}
\newcommand{\dd}{\mathrm{d}}
\newcommand{\mm}{\mathrm{m}}
\newcommand{\bp}{\mathbf{p}}
\newcommand{\bOmega}{\mathbf{\Omega}}

\newcommand{\tp}{\mathsf{T}}
\newcommand{\argmin}{\operatorname{argmin}}
\newcommand{\sign}{\operatorname{sgn}}
\newcommand{\proj}{\operatorname{proj}}
\newcommand{\bd}{\mathbf{bd\,}}
\newcommand{\interior}{\mathbf{int\,}}
\newcommand{\exit}{\mathrm{exit}}
\newcommand{\cnt}{\mathrm{cnt}}
\newcommand{\norm}[1]{\left\Vert #1 \right\Vert}
\newcommand{\abs}[1]{\left\vert #1 \right\vert}
\renewcommand{\tilde}{\widetilde}
\SetKw{KwOr}{or}
\SetKw{KwInit}{Initialize}
\SetKw{KwGoto}{goto}
\SetKw{KwReturn}{return}
\SetKwBlock{KwFna}{function \textnormal{listen($\exit_i, p_i^{(k)}$)}}{} 
\SetKwBlock{KwFnb}{function \textnormal{attack\_detection()}}{}

\begin{document}

\title{\huge{Autonomous and Resilient Control for Optimal LEO Satellite Constellation Coverage Against Space Threats}}

\author{Yuhan~Zhao
        and~Quanyan~Zhu%
\thanks{Yuhan~Zhao and Quanyan~Zhu are with the Department of Electrical and Computer Engineering, New York University, Brooklyn, NY, 11201 USA. E-mail: \{yhzhao, qz494\}@nyu.edu.}%
}

\maketitle
\thispagestyle{empty}

\begin{abstract}
LEO satellite constellation coverage has served as the base platform for various space applications. However, the rapidly evolving security environment such as orbit debris and adversarial space threats are greatly endangering the security of satellite constellation and integrity of the satellite constellation coverage. As on-orbit repairs are challenging, a distributed and autonomous protection mechanism is necessary to ensure the adaptation and self-healing of the satellite constellation coverage from different attacks. 
To this end, we establish an integrative and distributed framework to enable resilient satellite constellation coverage planning and control in a single orbit. Each satellite can make decisions individually to recover from adversarial and non-adversarial attacks and keep providing coverage service. We first provide models and methodologies to measure the coverage performance. Then, we formulate the joint resilient coverage planning-control problem as a two-stage problem. A coverage game is proposed to find the equilibrium constellation deployment for resilient coverage planning and an agent-based algorithm is developed to compute the equilibrium. The multi-waypoint Model Predictive Control (MPC) methodology is adopted to achieve autonomous self-healing control. Finally, we use a typical LEO satellite constellation as a case study to corroborate the results.
\end{abstract}



\section{Introduction} \label{sec:intro}
Recent advances in space technology research and development have inspired considerable applications for Low Earth Orbit (LEO) satellite constellations, such as positioning \cite{xu2007gps}, communications \cite{roddy2006satellite}, and remote sensing \cite{campbell2011introduction}. Among all the research and applications, satellite constellation coverage plays a fundamental role because it serves as the base platform for other space applications. 
In the satellite constellation coverage, multiple LEO satellites work cooperatively to provide global or regional coverage service to the ground \cite{ulybyshev2008satellite,lee2020satellite,al2021optimal}. 
A broad coverage can enable not only critical military operations such as strategic guidance and real-time surveillance for the ground agents but also civil applications such as satellite Internet in remote areas. 
\textcolor{black}{It also lays the foundation of future 6G communication technologies \cite{zhen2020energy,giordani2020non,yaacoub2020key}.}
Due to its prominent importance, various models and algorithms \cite{al2021analytic,okati2020downlink,savitri2017satellite,hitomi2018constellation} have focused on the optimal satellite constellation design and deployment to maximize the joint coverage. 

However, as space systems become more critical in different fields, growing security threats in the space domain make satellites more vulnerable during the operation, jeopardizing the coverage performance. For example, physical attacks \cite{liu2020space} such as missiles and laser attacks can directly destroy satellite entities. Cyber attacks such as jamming \cite{reiffen1964parametric,rausch2006jamming} can disrupt and block communication channels, causing degeneration or failure to the coverage service. It is estimated in \cite{manulis2020cyber} that there have been more than one hundred satellites attacked since 1997, including jamming and hijacking, many of which have caused significant consequences due to losses of navigation and communication. In addition to adversarial attacks, non-adversarial attacks such as space debris can also put satellites at risk. For example, in 2009, the collision of two communication satellites created more than 2000 pieces of debris \cite{collision_2009}, posing a threat to a large number of satellites in orbit. 
The hostile and intelligent attacks and the environmental hazard have led to a growing concern for space applications.

\begin{figure}
    \centering
    \includegraphics[scale=0.36]{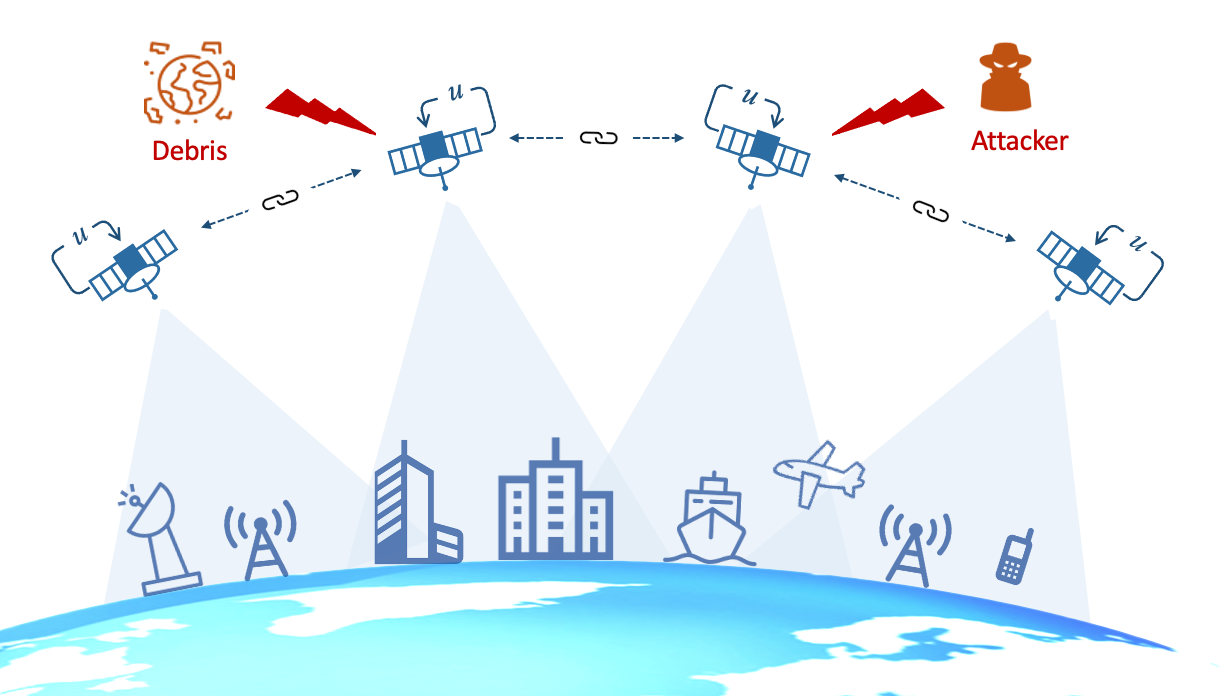}
    \caption{Distributed and resilient satellite constellation coverage framework: Intelligent attackers and space debris can hurt individual satellite's coverage capability. Each satellite needs to communicate only with neighbors to adapt to attacks and achieve autonomous self-healing control.}
    \label{fig:illustration}
\end{figure}

To address security challenges, we need a reliable mechanism to protect satellite constellations and improve the resilience of the satellite constellation coverage. Resilience not only ensures the survivability of satellite constellations under successful attacks but also provides flexibility and adaptability to cope with space security threats. 
However, the current solutions are not sufficient.
On the one hand, on-orbit repairs are considerably costly and challenging to cope with potential attacks. On the other hand, centralized approaches for optimal LEO satellite constellation designs such as \cite{lee2020satellite,savitri2017satellite} are also not sufficient to address security challenges for the following reasons. First, many constellation designs do not take satellite control into consideration, which makes it more challenging for the satellite constellation to adapt to new attacks. Second, since satellites scatter in space, they cannot connect to the same ground station simultaneously, making their coordination difficult.

The distributed architecture can fulfill the requirement for real-time attack adaptation and provide more flexibility to cope with environmental changes, which increases the chance of survivability.
In this work, we develop a distributed control framework that enables a resilient satellite constellation coverage planning and control in a single orbit, \textcolor{black}{so that every satellite can make individual decisions to mitigate both adversarial and non-adversarial attacks.}
More specifically, We first propose the average coverage cost to measure the coverage performance of a single-orbit satellite constellation. Next, we formulate a two-stage planning-control problem to jointly provide resilient and distributed coverage planning and self-healing control. At the planning stage, we propose a coverage game among all satellites to find the equilibrium (also local optimal) coverage deployment under various attack scenarios. At the control stage, we design multi-waypoint model predictive control to achieve autonomous self-healing for the satellite constellation. We also use a case study to demonstrate that our framework provides resilience to the satellite constellation coverage problem.

The contribution of this paper is summarized as follows.
\begin{itemize}
    \item We establish an integrative and distributed framework that enables the resilient and autonomous satellite constellation coverage planning-control against various space security threats.
    \item We propose a coverage game to find the equilibrium (also local optimal) satellite constellation deployment after security attacks. We also develop an agent-based iterative algorithm to compute the equilibrium of the coverage game. 
    \item We provide a thorough analysis of the optimality property of the coverage game equilibrium generated by our algorithm.  
    \item We adopt multi-waypoint model predictive control to achieve autonomous self-healing control to actively adapt to adversarial and non-adversarial space security attacks.
\end{itemize}

\subsection{Related Work}
Most research in the satellite domain studies the satellite coverage design and the satellite control separately. We discuss the related work in these two areas as well as some works in resilient control and its applications.

\subsubsection{Satellite Coverage Design}
The research in LEO satellite constellation coverage can be roughly divided into global coverage design and regional coverage design. Global coverage design focuses on designing satellite constellations to provide continuous global coverage. Classical approaches include Walker constellation \cite{walker1984satellite,walker1970circular} and street-of-coverage methods \cite{luders1961satellite,rider1986analytic,beste1978design}. The Walker constellation specifies circular orbits with the same altitude. Different orbits are distinguished by the orbital inclinations. The orbital inclinations and satellite deployments are selected for the largest coverage. 
In the street-of-coverage method, each orbit provides a strip of coverage region. By computing the orbital inclination, the method determines the least orbits which are needed to maximize the overall coverage. Lang and Adams in \cite{lang1998comparison} have provided a review and comparison between Walker constellation coverage and street-of-coverage. 
In addition to classical approaches, elliptic constellations are also used to provide the global coverage \cite{draim1987common,mortari2004flower}. Draim in \cite{draim1987common} has shown that with only 4 elliptic satellites, the global and continuous line-of-sight coverage can be achieved. 
Al-Hourani in \cite{al2021optimal} has proposed an analytic framework based on stochastic geometry to determined the optimal orbital altitude for the maximal satellite constellation coverage.

Regional satellite constellations only cover specific regions on the earth. For example, the Indian Regional Navigation Satellite System (IRNSS) \cite{irnss} and the Quasi-Zenith Satellite System (QZSS) \cite{qzss}. The regional constellation design with repeating ground track orbits has been studied in \cite{hanson1990designing}. 
Wang et al. in \cite{wang2008optimization} have proposed a genetic algorithm to maximize the regional coverage of the reconnaissance satellite constellation and minimize the number of satellites. 
Meziane-Tani et al. in \cite{meziane2016optimization} have designed the regional coverage satellite constellation with reduced constellation size by using the evolutionary optimization method. 
Lee and Ho in \cite{lee2020satellite} have approached the regional constellation design by proposing the circular convolution formulation and using binary integer linear programming to select the optimal constellation pattern.

\subsubsection{Distributed Satellite Control}

In the research of multi-satellite control, the notions of ``satellite swarm" and ``satellite cluster" are more involved. Many distributed control algorithms have been studied for satellite swarms and satellite clusters. Wang et al. in \cite{wang2019self} have developed a distributed algorithm based on potential field and satellite relative dynamics to control a satellite swarm.
Izzo and Pettazzi in \cite{izzo2007autonomous} have exploited the behavior-based approach for autonomous and distributed path planning of a satellite swarm.
Works such as \cite{nag2013behaviour,massioni2010decomposition} have also developed feasible distributed control methods for satellite swarms.

\subsubsection{Resilient Control}
Although few works have focused on resilient control in the space domain, it has been studied in other fields with different approaches. For example, in cyber-physical systems, Zhu and Ba\c{s}ar have proposed a game-theoretic framework in \cite{zhu2015game} to cope with potential cyber attacks and maintain the system performance. The security and resilience issues of cyber-physical systems have been extensively investigated in \cite{zhu2020crosslayer}. In network systems, a dynamical game framework has been proposed in \cite{chen2019dynamic} to perform the resilient network design and control for the connectivity of infrastructure network systems. The transactive resilience of microgrid systems has been discussed in \cite{chen2021transactive} through a contract-theoretic approach. The monograph \cite{chen2020agame} has thoroughly investigated the resilient analysis and design for independent network systems. In multi-agent systems, Chen and Zhu in \cite{chen2019control} have studied the resilient connectivity control of multi-robot systems with game theory. Huang et al. in \cite{huang2020dynamic} have also provided a comprehensive review on the robust and resilient design and control with dynamic game theory.

\subsection{Organization of the Paper}
The rest of the paper is organized as follows. Section \ref{sec:model} presents models and methodologies to measure the coverage performance and the formulation of the two-stage control-planning problem. Section \ref{sec:coverage_planning} discusses the coverage game for resilient planning, the agent-based algorithm the equilibrium deployment, and the related analysis. We consolidate our distributed planning-control framework in Section \ref{sec:synthesis}. Section \ref{sec:case_study} demonstrate the resilience of our framework with four case studies. Section \ref{sec:conclusion} concludes the paper.

\subsection{Notations}
We use $n$ as the total number of satellites and $\mN = \{1,,\dots, n\}$ as the set of all satellites; $r_s$ is the orbital radius; $\omega$ is the satellite mean motion; $T_s = 2\pi / \omega$ is the period.  For satellite $i \in \mN$, $\mN_i$ represents the neighbor set; $\alpha_i \in \R$ is the coverage angle; $\psi^\mm_i \in \R$ is the maximum coverage intensity. For satellite controls, $T_f \in \R$ denotes the control horizon; $u^\mm_i \in \R$ is the the maximum thrust-to-mass ratio; $p^\mm \in \R$ is the maximum deviation distance for CW equations to hold.
We use $p_i$ as the relative position vector and write $p_{ix}, p_{iy}$ as its scalar components along the $x$-axis and $y$-axis. We use bold $\bp$ as the aggregated vector, i.e., $\bp = \{p_1, \dots, p_n\} \in \R^{\sum_i \dim p_i}$. 
The notation $-i$ represents all satellites except for satellite $i$. For example, $p_{-i}$ represents $\{p_1, \dots, p_{i-1}, p_{i+1}, \dots, p_n \}$. Sometimes we write $\{p_i, p_{-i}\}$ instead of $\bp$ to emphasize the role of $p_i$ in $\bp$.
We use the weighted norm $\norm{x}_A = \sqrt{x^\tp A x}$ and the row-vector partial derivative $\frac{\partial f}{\partial x}$.

\section{System Model} \label{sec:model}
In this section, we first introduce the satellite dynamic model for control. Then we introduce the metric to measure the satellite constellation coverage performance. Finally, we formulate the coverage planning-control problem as a two-stage problem for a resilient and distributed architecture.

\subsection{Satellite Relative Dynamics}
The relative dynamics characterizes the motion of a (deputy) satellite with respect to another (chief) satellite \cite{he2008dynamics,tong2007relative}.
We consider LEO satellites that operate in circular orbits due to the small orbit eccentricity. We attach a moving frame $S_i$-$xyz$ called the \emph{local vertical local horizontal} (LVLH) frame to satellite $i$ shown in Fig.~\ref{fig:sat_lvlh}, where $S_i$ is the chief position, $x$-axis points outward along the radial direction, $y$-axis points to the velocity direction, and $z$-axis is perpendicular the orbital plane. When satellite $i$ deviates from its chief position, the chief becomes virtual, and its motion can be captured by the Clohessy-Wiltshire (CW) equations \cite{clohessy1960terminal} if the deviation is sufficiently small compared with the orbital radius:  
\begin{align}
    \delta \ddot{x} - 3\omega^2 \delta x - 2\omega \delta\dot{y} & = u_x, \label{eq:cw.1} \\
    \delta \ddot{y} + 2\omega \delta\dot{x} & = u_y, \label{eq:cw.2} \\ 
    \delta \ddot{z} + \omega^2 \delta z &= u_z, \label{eq:cw.3}
\end{align}
where $(\delta x, \delta y, \delta z) \in \R^3$ and $(u_x, u_y, y_z) \in \R^3$ represent the displacements and external thrusts along $x,y,z$-axis respectively. The CW equations have been adopted for satellite control in different fields \cite{stupik2012optimal,sullivan2017comprehensive,ye2020satellite,gong2020pursuit} such as satellite pursuit-evasion. 
We note that the motion along $z$-axis is independent from the one in the orbital plane ($xy$-plane). The satellite can remain static in $z$-axis if we use zero control $u_z = 0$ and zero initial condition $\delta z = \delta \dot{z} = 0$. In this work, we constrain satellites to operate in the same orbit plane. Therefore, we focus on satellite controls in $xy$-plane and ignore \eqref{eq:cw.3}. Let $p_i = [p_{ix} \ p_{iy}] \in \R^2$ and $v_i = [v_{ix} \ v_{iy}] \in \R^2$ denote the relative position and relative velocity of satellite $i$ in the LVLH frame, respectively. Let $u_i = [u_{ix} \ u_{iy}] \in \R^2$ be the external control. We arrive at
\begin{equation}
\small
    \label{eq:control}
    \begin{bmatrix} \dot{p}_i \\ \dot{v}_i \end{bmatrix} = 
    \begin{bmatrix} 0 & 0 & 1 & 0 \\ 0 & 0 & 0 & 1 \\ 3\omega^2 & 0 & 0 & 2\omega \\ 0 & 0 & -2\omega & 0 \end{bmatrix} \begin{bmatrix} p_i \\ v_i \end{bmatrix} + 
    \begin{bmatrix} 0_{2\times 1} \\ u_i \end{bmatrix} := A \begin{bmatrix} p_i \\ v_i \end{bmatrix} + B u_i,
\normalsize
\end{equation}

\begin{remark}
One reason to keep satellites in the same orbit is to focus on the collective coverage behavior of the satellite constellation. When some satellite jumps out of the current orbit plane, it will form a new orbit to operate. This will the jeopardize entire constellation and hence affect the constellation coverage performance. Satellite coordination becomes more challenging when satellites are in different orbits. 
\end{remark}

\begin{remark}
In practice, there are external disturbances such as gravitational perturbations from other satellites, which can deviate the satellite from the current orbit plane. However, since $z$-axis motion can be controlled independently, in case of perturbations, we can simply deploy another controller to stabilize $z$-axis motion and keep the satellite in the orbit plane.
\end{remark}

CW equations rely on linearization. They become less accurate when the satellite is too far away from its chief position. We can use the constraint $\norm{p_i}_2 \leq p^\mm$ to ensure the accuracy of the satellite dynamics \eqref{eq:control}. 

When satellite $i$ moves to the relative position $p_i$, it forms a deviation angle $\Delta \phi_i(p_i)$ with $x$-axis of the LVLH frame shown in Fig.~\ref{fig:sat_lvlh}. Using geometry, we can compute
\begin{equation}
\label{eq:Dphi}
    \Delta \phi_i(p_i) = \sign(p_{iy}) \arccos \left( \frac{p_{ix}+r_s}{ \sqrt{(p_{ix}+r_s)^2 + p_{iy}^2} } \right).
\normalsize
\end{equation}

\begin{figure}
    \centering
    \includegraphics[height=3.2cm]{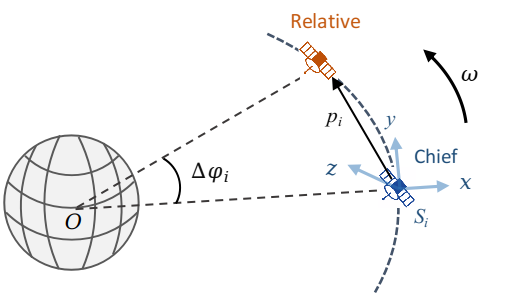}
    \caption{The LVLH frame of satellite $i$ is attached on the chief (blue) position. The deputy (red) position is characterized by the relative position $p_i$.}
    \label{fig:sat_lvlh}
\end{figure}

\subsection{Satellite Configuration}
We consider $n$ satellite rotating in the same circular orbit numbered from $1$ to $n$. The index increasing direction is the same as the satellite moving direction. 

\begin{definition}
A satellite configuration (or configuration) of $n$ LEO satellites in the same orbit refers to a stable formation such that the relative position between any two satellites remains constant.
\end{definition}

In the configuration, all satellites rotate passively in the orbit; i.e., the satellites always maintain the same formation and do not use external thrusts. The only difference between two configurations are relative positions. Therefore, we can measure one configuration with respect to a particular configuration called the \emph{initial configuration} (IC). We set the satellite positions in the IC as the chief positions. Then the position of satellite $i$ in a new configuration can be characterized by the relative position $p_i$. Therefore, given an IC, any new configuration can be characterized the relative position vector $\bp = \{ p_1, \dots, p_n \} \in \R^{2n}$. In particular, $\bp^{\mathrm{IC}} = 0_{2n\times 1}$. Also from \eqref{eq:Dphi}, we have $\Delta \phi_i(p_i) = 0$ in the IC for every $i \in \mN$.

On the other hand, satellites are not stationary to the earth. To characterize satellite positions with respect to the ground, we first introduce the geocentric polar frame $O$-$L$ shown in Fig.~\ref{fig:sat_configuration}, which is fixed on the earth's surface. Then we define the configuration angle.

\begin{definition}
The configuration angle $\phi_i$ for satellite $i$ is the angle between $x$-axis of the LVLH frame and $OL$ axis, which changes with measuring time $\tau$ and the relative position $p_i$:
\begin{equation}
    \label{eq:phi_i}
    \phi_i(p_i, \tau) = \phi_i^0 + \Delta \phi_i(p_i) + \omega \tau,
\end{equation}
where $\phi^0_i$ refers to the initial angle when $\tau = 0$. 
\end{definition}

In fact, since the satellite motion is is periodic, we can measure the satellite configuration angle in the $O$-$L$ frame starting from any time $\tau$. For simplicity, we set $\tau=0$ when $x$-axis of satellite $1$ coincides with $OL$-axis. In this case, we always set $\phi^0_1 = 0$. 

\begin{remark}
Satellites are not stationary relative to the ground but they provide coverage service to the ground. We need configuration angles to characterize the coverage performance of individual and group satellites later. For the configuration measurement, the vector $\bp$ with the IC is sufficient.
\end{remark}

\begin{figure}
    \centering
    \includegraphics[height=4cm]{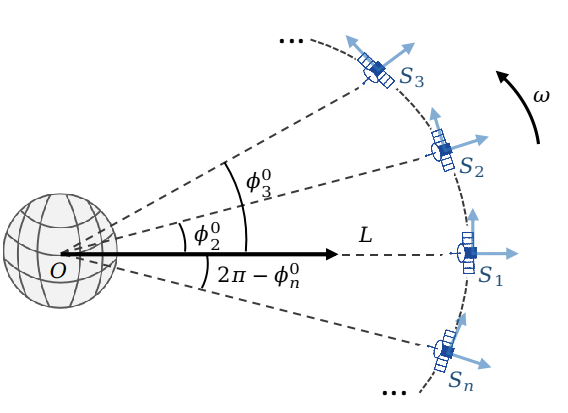}
    \caption{Satellite configuration at $\tau = 0$.}
    \label{fig:sat_configuration}
\end{figure}

\subsection{Coverage Measure}
After forming a new configuration, satellites begin to provide stable coverage service. Let $\mu : \R \to \R$ be the demand intensity on the orbit ground track measured in $O$-$L$ frame. Each surface point corresponds to a demand intensity, which can be obtained by data in practice.
Due to the short period of LEO satellites, we assume that $\mu$ is time-invariant. For the computational purpose, we extend $\mu(\theta)$ to a periodic function with period $2\pi$, i.e., $\mu(\theta) = \mu(\theta + 2k\pi)$, $k \in \Z$. 

At time $\tau \in \R$, satellite $i$ covers part of the earth surface shown in Fig.~\ref{fig:coverage_geometr}. The coverage angle $\alpha_i$ can be computed by field of view (FOV) angle and the coverage geometry. We denote $C_i(p_i, \tau)$ as the coverage region (which is an interval) and assume symmetric coverage to the ground. Then, we have
\begin{equation*}
    C_i := C_i(p_i, \tau) = C_i^+ \cup C_i^- \cup \{ \phi_i(p_i, \tau) \},
\end{equation*}
where
\begin{equation}
    \label{eq:ci}
    \begin{split}
        C_i^+ &:= ( \phi_i(p_i, \tau), \phi_i(p_i,\tau) + \alpha_i ), \\ 
        C_i^- &:= ( \phi_i(p_i, \tau) - \alpha_i, \phi_i(p_i,\tau) ).
    \end{split}
\end{equation}
Since a configuration is an end-to-end formation, we define $C_i$ within the interval $[\omega \tau, \omega \tau + 2\pi)$ at time $\tau$ for mathematical characterization. It may create discontinuities for $C_i$. For example, $C_1(0,0) = [0, \alpha_1) \cup (2\pi-\alpha_1, 2\pi)$. However, it can be circumvented by shifting the discontinued region left or right by $2\pi$.

Note that $C_i$ may overlap with $C_j$ ($i,j \in \mN, j\neq i$) to ensure the full coverage of the ground. Due to the ring structure of the configuration, we assume that satellite $i$ can only overlap with its adjacent neighbors, i.e., $\mN_i = \{ i-1, i+1\}$\footnote{For clarity, the indices are cyclic with a modulus $n$.}. Non-overlap scenarios yield $\mN_i = \varnothing$. 

\begin{remark}
The assumption of overlapping with adjacent neighbors can be justified by economic and practical reasons. Deploying multiple satellites to cover the same area in the same orbit can be costly and inefficient. The satellite control and collision avoidance also become more challenging. 
\end{remark}

A satellite may provide different coverage intensities in $C_i$. We adopt the linear coverage intensity function. Let $k_i := \psi^\mm_i / \alpha_i$, we define the local coverage intensity $\psi_i: \R^2 \times \R \times \R \to \R$ by
\begin{equation}
    \small
    \label{eq:psi_i}
    \begin{split}
        & \psi_i(p_i, \theta, \tau) = \\
        &\quad \begin{cases} 
        -k_i(\theta - \omega \tau - \Delta \phi_i(p_i) - \phi^0_i) + \psi^\mm_i & \theta \in C_i^+ \cup \{\phi_i\} \\ k_i(\theta - \omega \tau - \Delta \phi_i(p_i) - \phi^0_i) + \psi^\mm_i  & \theta \in C_i^- \\
        0 & \text{o.w.}
    \end{cases}
    \end{split}
\normalsize
\end{equation}
When $C_i$ is not continuous in $[\omega \tau, \omega \tau+2\pi)$, $\theta$ in \eqref{eq:psi_i} needs to shift left or right by $2\pi$, depending on the position of satellite $i$. We omit the special definitions of $\psi(\theta, p_i, \tau)$ that need shift for simplicity.

In the configuration, satellite $i$ can affect the coverage intensity in region $\tilde{C}_i := C_{i-1}^+ \cup C_i \cup C_{i+1}^-$ because of potential overlaps with adjacent neighbors. We define the composite coverage intensity $\beta_i: \R^6 \times \R \times \R \to \R$ for satellite $i$ in $\tilde{C}_i$ as $\beta_i( \{p_i, p_{i-1}, p_{i+1}\}, \theta, \tau) = \sum_{i \in \mN_i} \psi_i(p_i, \theta, \tau)$. Although $\beta_i$ only relates to $p_{i-1}$ and $p_{i+1}$, we can extend its arguments by writing $\beta_i(\{p_i, p_{-i}\}, \theta, \tau)=: \beta_i(\bp, \theta, \tau) $ for simplicity of notations.
Then, we define the average coverage cost $u_i: \R^2 \times \R^{2n-2} \to \R$ for satellite $i$ as 
\begin{equation}
\label{eq:ui}
    u_i(p_i, p_{-i}) = 
    \frac{1}{2 T_s} \int_0^{T_s} \int_{\tilde{C}_i} \norm{\beta_i(\bp,\theta,\tau) - \mu(\theta)}^2_2 \dd\theta \dd\tau.
\end{equation}
The cost \eqref{eq:ui} measures the average coverage performance of satellite $i$ over a single period. A smaller cost indicates a better coverage performance. Also, satellite $i$ should decide its new position  $p_i$ from the feasible set $\Omega_i = \{ p_i \in \R^2: \norm{p_i} \leq p^{\mm}, (p_{ix}+r_s)^2 + p_{iy}^2 = r_s^2 \}$, i.e. $p_i \in \Omega_i$, which indicates that the satellite should remain in the same orbit and the relative motion should not be far from its chief position.

\begin{figure}
    \centering
    \includegraphics[height=2.4cm]{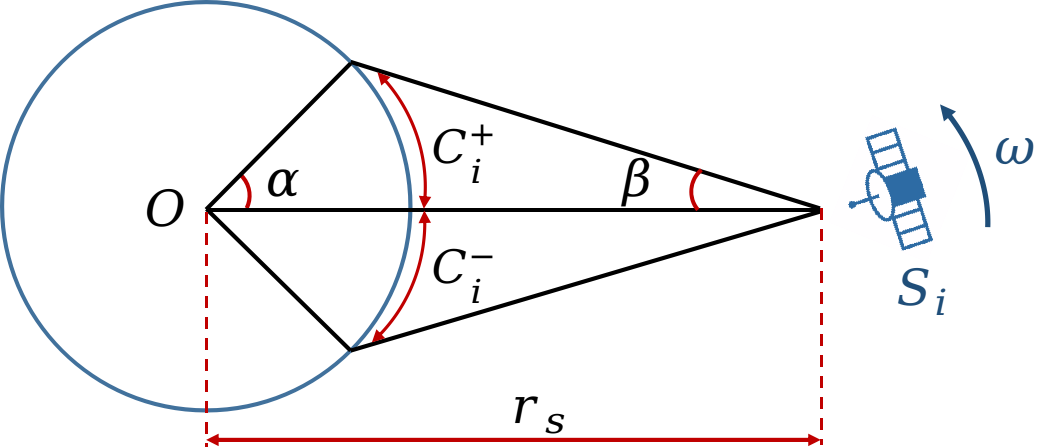}
    \caption{Coverage geometry of satellite $i$. Here, $\alpha$ and $\beta$ represent the coverage angle and FOV angle, respectively; $r_s$ is the orbital radius; $C_i^+$ and $C_i^-$ are two half-coverage regions.}
    \label{fig:coverage_geometr}
\end{figure}

\subsection{Two-Stage Problem Formulation}
When unexpected adversarial or non-adversarial attacks occur, some satellites' coverage capability can be affected. The current configuration may no longer provide the optimal coverage service. To cope with coverage performance degeneration, we formulate a two-stage problem---the planning and control stages---as a resilient and distributed architecture to (a) improve the resilience of the satellite configuration to adapt to space threats and (b) optimize the fuel consumption (the control effort) for maneuver by considering limited thrusts and fuel supply.

The planning stage problem seeks a new configuration when satellites encounter incidents and fail to provide the optimal coverage with the current configuration. More specifically, every satellite minimizes its cost objective defined in \eqref{eq:ui} and hence form a non-cooperative game called the \emph{coverage game} $\mG$, which can be written in a strategic form $\mG = \langle n, (\Omega_i)_{i=1}^n, (u_i)_{i=1}^n \rangle$. The equilibrium solution of the coverage game $\mG$ will be used as the new configuration to adapt to the coverage performance degradation. 

Let $\bp^\dd$ be the new configuration generated by the coverage game $\mG$. The control stage problem steers all satellites to $\bp^\dd$ by minimizing fuel consumptions. Let $p^\dd_i$ be the $i$-th component of $\bp^\dd$. Due to independent dynamics \eqref{eq:control}, each satellite can autonomously drive to the target position by using its own fuel-optimal controls. For satellite $i$, the control stage problem can be formulated as a finite-time optimal control problem:
\begin{equation}
    \label{eq:Qci}
    \tag{$\mQ_{ci}$}
    \begin{split}
        \min_{u_i} \quad & \frac{1}{2} \int_{t=0}^{T_f} \left( \norm{u_i}_{R_i}^2 + \norm{v_i}_{Q_i}^2 \right) \dd t \\
        \text{s.t.} \quad & \begin{bmatrix} \dot{p}_i \\ \dot{v}_i \end{bmatrix} = A \begin{bmatrix} p_i \\ v_i \end{bmatrix} + B u_i, \quad \norm{u_i(t)}_2 \leq u_i^\mm, \\
        & p_{i}(T_f) = p_i^\dd, \quad v_i(T_f) = 0.
    \end{split}
\end{equation}
The terminal constraints require that all satellites indeed form the desired stable configuration after the control.

\section{Distributed Coverage Planning and Analysis} \label{sec:coverage_planning}
In this section, we discuss the coverage game as an approach to distributed and resilient coverage planning. We first study the Nash equilibrium of the coverage game and then introduce an agent-based (also distributed) algorithm for the coverage planning. We also analyze the property of the equilibrium solution generated by our algorithm.   

\subsection{Nash Equilibrium of Coverage Game}
In the coverage game $\mG = \langle n, (\Omega_i)_{i=1}^n, (u_i)_{i=1}^n \rangle$, $n$ satellites are the players and $p_i$ is called a pure strategy of satellite $i$. The feasible set $\Omega_i$ is also called the action space. We write $\bOmega = \prod_{i=1}^n \Omega_i$ to denote the set of all strategies and $\Omega_{-i} = \prod_{j=1,j\neq i}^n \Omega_j$.
We use Nash equilibrium (NE) as the solution concept to study the coverage game $\mG$. 

\begin{definition}
\label{def:ne}
A strategy $\bp \in \bOmega$ is a Nash equilibrium if
\begin{equation*}
    u_i(p_i, p_{-i}) \leq u_i(p_i', p_{-i}), \quad \forall p_i' \in \Omega_i, \ \forall i \in \mN.
\end{equation*}
A strategy $\bp \in \prod_{i=1}^n \mB_i$ is a local Nash equilibrium \cite{ratliff2016characterization} if there exists open sets $\mB_i \subset \Omega_i$ containing $p_i$ for every $i \in \mN$, such that 
\begin{equation*}
    u_i(p_i, p_{-i}) \leq u_i(p_i', p_{-i}), \quad \forall p_i' \in \mB_i, \ \forall i \in \mN.
\end{equation*}
\end{definition}

To study the NE of the coverage game $\mG$, we first note that $\mG$ is closely related to a special class of games called \emph{potential games} \cite{monderer1996potential}. We provide the definition of the potential game.

\begin{definition}
\label{def:potential}
A game $\langle n, (\Omega_i)_{i=1}^n, (u_i)_{i=1}^n \rangle$ is called a potential game if there exists a potential function $J(\bp)$ such that for every $i \in \mN$ and every $p_{-i} \in \Omega_{-i}$, 
\begin{equation*}
    u_i(p_i'', p_{-i}) - u_i(p_i', p_{-i}) = J( \{p_i'', p_{-i}\} ) - J( \{p_i', p_{-i}\} )
\end{equation*}
for all $p_i', p_i'' \in \Omega_i$.
\end{definition}

Then we have the following proposition to characterize $\mG$.
\begin{proposition}
\label{prop:1}
The coverage game $\mG$ is a potential game.
\end{proposition}

\begin{proof}
We first define the global coverage intensity $\rho: \R^{2n} \times \R \times \R \to \R$ in $[\omega \tau, \omega \tau + 2\pi)$ by $\rho(\bp, \theta, \tau) = \sum_{i=1}^n \psi_i(p_i, \theta, \tau)$.
Then we define the function $J: \R^{2n} \to \R$ by
\begin{equation}
\label{eq:cost}
    J(\bp) = \frac{1}{2T_s} \int_0^{T_s} \int_{\omega \tau}^{\omega \tau + 2\pi} \norm{\rho(\bp, \theta, \tau) - \mu(\theta)}^2_2 \dd\theta \dd\tau.    
\end{equation}

For simplicity, we write $S := [\omega \tau, \omega \tau + 2\pi)$. For any $i \in \mN$, let $p_i', p_i'' \in \Omega_i$. We have 
\begin{align*}
    & J(\{p_i', p_{-i}\}) - J(\{p_i'', p_{-i}\}) \\ 
    =\ & \frac{1}{2T_s} \int_0^{T_s} \int_{\tilde{C}_i \cup (S\backslash\tilde{C}_i)} \left( \norm{\rho(\{p_i', p_{-i}\}, \theta, \tau) - \mu(\theta)}^2_2 \right. \\ 
    & \hspace{3em} \left. - \norm{\rho(\{p_i'', p_{-i}\},\theta, \tau) - \mu(\theta)}^2_2 \right) \dd \theta \dd \tau.
\end{align*}
We split the inner integration into two halves. From the assumption in Sec.~\ref{sec:model}, $p_i$ has no impact on the coverage intensity defined outside $\tilde{C}_i$. i.e., $\rho(\{p_i', p_{-i}\}, \theta, \tau) = \rho(\{p_i'', p_{-i}\}, \theta, \tau)$ for $\theta \in (S\backslash \tilde{C}_i)$. Therefore, we can cancel the integral in $(S\backslash \tilde{C}_i)$ and leave the integral in $\tilde{C}_i$. 
From the definitions of $\rho$ and $\beta_i$, we can check that $\rho(\{p_i, p_{-i}\}, \theta, \tau)$ is equivalent to $\beta_i(\{p_i, p_{-i}\}, \theta, \tau)$ in $\tilde{C}_i$, which indicates
\begin{gather*}
    \int_0^{T_s} \int_{\tilde{C}_i} \norm{\rho(\{p_i', p_{-i}\}, \theta, \tau) - \mu(\theta)}^2_2 \dd\theta \dd\tau \\ 
    = \int_0^{T_s} \int_{\tilde{C}_i} \norm{\beta_i(\{p_i', p_{-i}\}, \theta, \tau) - \mu(\theta)}^2_2 \dd\theta \dd\tau 
\end{gather*}
for all $p_i' \in \Omega_i$. Therefore, we can obtain
\begin{equation*}
    u_i(p_i'', p_{-i}) - u_i(p_i', p_{-i}) = J(\{p_i'', p_{-i}\}) - J(\{p_i', p_{-i}\}),
\end{equation*}
for all $i \in \mN$, which completes the proof.
\end{proof}

\subsubsection{Interpretation of Potential Function}
Unlike other potential games, the potential function $J$ of the coverage game $\mG$ has a clear interpretation. It measures the coverage performance of the entire satellite configuration, and we refer to it as the accumulated average coverage cost. It can also be used as an index to check the quality of the local NE. As the definition \eqref{eq:cost} indicates, a smaller $J$ indicates a better coverage performance.

\subsection{Existence of NE}
From Def.~\ref{def:ne}, an NE $\bp$ of $\mG$ indicates $u_i(p_i, p_{-i}) - u_i(p_i', p_{-i}) \leq 0$ for all $p_i' \in \Omega_i$, $i\in \mN$. Since $\mG$ is a potential game, we can use the potential function to substitute the satellite cost function and obtain
\begin{equation*}
    J(\{p_i, p_{-i}\}) - J(\{p_i', p_{-i}\}) \leq 0, \quad \forall p_i' \in \Omega_i, \ \forall i \in \mN.
\end{equation*}
This shows that the minimizer of $J(\bp)$ is also an NE. Therefore, the argmin set of the potential function $J(\bp)$ is a subset of NE of the coverage game $\mG$. We can use this property to characterize the existence of the NE and arrive at the following proposition. 

\begin{proposition}
\label{prop:ne_exist}
The NE of the coverage game $\mG$ exists.
\end{proposition}

\begin{proof}
It suffices to show the set $\argmin_{\bp \in \bOmega} J(\bp)$ exists. It is clear that $\Omega_i$ is a compact set and so is $\bOmega$. From the definition of the potential function in \eqref{eq:cost}, we can check that $J(\bp)$ is continuous in $\R^{2n}$. From the Extreme Value Theorem, there exists $\bp^* \in \bOmega$ such that $J(\bp^*) = \inf\{ J(\bp), \bp \in \bOmega \}$. This shows the existence of NE of $\mG$. 
\end{proof}

\begin{remark}
In the following sections, we suppress function arguments for simplicity. For example, $\rho$ stands for $\rho(\bp, \theta, \tau)$. 
We only write the argument when it is emphasized. 
\end{remark}

\subsection{Agent-Based Algorithm for NE Computation}
Potential games provide a way to compute the NE of the coverage $\mG$ by solving the following optimization problem:
\begin{equation}
    \label{eq:Qp}
    \tag{$\mQ_p$}
    \min_{\bp \in \bOmega} \quad J(\bp).
\end{equation}
Due to the nonconvexity of $\bOmega$, we look for the local minimum of $J$, which is also the local NE of the coverage game $\mG$. The projected gradient descent method can be used to search for the local minimum. 
In the following, we will show that the gradient method can be decentralized so that we can devise agent-based (also distributed) algorithms to compute the local minimum of $J$.

\subsubsection{Distributed Structure of Coverage Measure}
Despite the coupling of $\bp$ in $J$, for satellite $i$, from \eqref{eq:psi_i} we have 
\begin{equation}
\small
    \label{eq:dpsi_dpi}
    \frac{\partial \psi_i}{\partial p_i} = \begin{cases} 
    k_i \left[\frac{-p_{iy}}{(p_{ix}+r_s)^2 + p_{iy}^2} \ \frac{p_{ix}+r_s}{(p_{ix}+r_s)^2 + p_{iy}^2} \right] & \theta \in C_i^+ \cup \{\phi_i\} \\ 
    -k_i \left[\frac{-p_{iy}}{(p_{ix}+r_s)^2 + p_{iy}^2} \ \frac{p_{ix}+r_s}{(p_{ix}+r_s)^2 + p_{iy}^2} \right] & \theta \in C_i^- \\ 
    0 & \text{o.w.}
    \end{cases}
\normalsize
\end{equation}
Following the definition of $\rho$, we have
\begin{equation}
\small
    \label{eq:dJ_dpi_tmp}
    \begin{split}
        &\frac{\partial J}{\partial p_i} = \frac{1}{2 T_s} \int_{\tau = 0}^{T_s} \int_{\theta = \omega \tau}^{\omega \tau+2\pi} 2 (\rho - \mu) \frac{\partial \rho}{\partial p_i} \dd \theta \dd \tau \\ 
        &= \frac{1}{2 T_s} \int_{\tau = 0}^{T_s} \int_{C_i} 2 (\rho - \mu) \frac{\partial \psi_i}{\partial p_i} \dd \theta \dd \tau \\
        &= \frac{1}{T_s} \int_{\tau = 0}^{T_s} \int_{C_i} \rho \frac{\partial \psi_i}{\partial p_i} \dd \theta \dd \tau - \frac{1}{T_s} \int_{\tau = 0}^{T_s} \int_{C_i} \mu \frac{\partial \psi_i}{\partial p_i} \dd \theta \dd \tau.
    \end{split}
\normalsize
\end{equation}
With the assumption of the time-invariant demand $\mu$, we can further simplify \eqref{eq:dJ_dpi_tmp} with the following lemma. 

\begin{lemma}\label{lemma:1}
Let $f:\R \to \R$ be a function with $f(x) \geq 0$ and $f(x) = f(x + T)$. For any $\delta \in [0, T)$, the function $G(x) = F(x + \delta) - F(x)$ is periodic with period $T$, where $F(x) = \int f(x) \dd x$. Hence $\int_0^T G(x) \dd x = \int_0^T G(x+\epsilon) \dd x$ for any $\epsilon \in \R$.
\end{lemma}

\begin{proof}
See Appendix \ref{app:1}.
\end{proof}

Using Lemma \ref{lemma:1}, we arrive at the following proposition.
\begin{proposition}
\label{prop:int_simplification}
The integral $\frac{1}{T_s} \int_0^{T_s} \int_{C_i} \mu \frac{\partial \psi_i}{\partial p_i} \dd \theta \dd \tau$ in \eqref{eq:dJ_dpi_tmp} is equal to $0$ for all satellite $i$, $i\in \mN$.
\end{proposition}

\begin{proof}
We divide the integral into two halves:
\begin{equation}
\label{eq:prpo3.1}
\small
    \frac{1}{T_s} \int_0^{T_s} \int_{C_i^+} \mu \frac{\partial \psi_i}{\partial p_i} \dd\theta \dd\tau + \frac{1}{T_s} \int_0^{T_s} \int_{C_i^-} \mu \frac{\partial \psi_i}{\partial p_i} \dd\theta \dd\tau.
\normalsize
\end{equation}
Let $M(\theta) = \int \mu(\theta) \dd\theta$. For any fixed $\tau$, we have
\begin{equation*}
\small
    \begin{split}
        \int_{C_i^+} \mu \frac{\partial \psi_i}{\partial p_i} \dd\theta &= k_i \frac{\partial \Delta \phi_i}{\partial p_i} \left( M(\phi_i + \alpha_i) - M(\phi_i) \right), \\
        \int_{C_i^-} \mu \frac{\partial \psi_i}{\partial p_i} \dd\theta &= -k_i \frac{\partial \Delta \phi_i}{\partial p_i} \left( M(\phi_i) - M(\phi_i - \alpha_i) \right).
    \end{split}
\normalsize
\end{equation*}

Note that $\mu(\theta) \geq 0$ and has a period $2\pi$, and it is clear that $\alpha_i < 2\pi$. Therefore, from Lemma.~\ref{lemma:1}, we know that $M(\phi_i + \alpha_i) - M(\phi_i)$ and $M(\phi_i) - M(\phi_i-\alpha_i)$ are both periodic with period $2\pi$. 
However, $\phi_i$ is a function of $\tau$ from \eqref{eq:phi_i}, and the variables $\theta, \tau$ are correlated in \eqref{eq:dJ_dpi_tmp} by $\theta = \omega \tau$. Therefore, $M(\phi_i + \alpha_i) - M(\phi_i)$ and $M(\phi_i) - M(\phi_i-\alpha_i)$ are in fact periodic in $\tau$ with period $T_s$. 
Let $G(\tau) = M(\phi_i + \alpha_i) - M(\phi_i)$, then $G(\tau)$ has the period $T_s$. The integral \eqref{eq:prpo3.1} becomes
\begin{equation*}
    \small
    \frac{1}{T_s} k_i \frac{\partial \Delta \phi_i}{\partial p_i} \left( \int_0^{T_s} G(\tau) \dd\tau - \int_0^{T_s} G(\tau-\alpha_i) \dd\tau \right) = 0,
\end{equation*}
\normalsize
which completes the proof.
\end{proof}

With Prop.~\ref{prop:int_simplification}, \eqref{eq:dJ_dpi_tmp} becomes 
\begin{equation}
    \label{eq:dJ_dpi}
    \frac{\partial J}{\partial p_i} = \frac{1}{T_s} \int_0^{T_s} \int_{C_i} \rho \frac{\partial \psi_i}{\partial p_i} \dd \theta \dd \tau. 
\end{equation}
From \eqref{eq:dJ_dpi} we observe that the satellite $i$ only needs to communicate with its adjacent neighbors to compute the necessary gradient information. Thus, we use distributed gradient descent (GD) methods for all satellites to jointly solve \eqref{eq:Qp}.

\subsubsection{Distributed Projected Gradient Descent Algorithm}
Due to the presence of $\Omega_i$, $i \in \mN$, we project the every iteration of the GD back to $\Omega_i$. The iteration follows
\begin{equation*}
\small
    p_i^{(k+1)} = \proj_{\Omega_i} \left( p_i^{(k)} - s^{(k)} \frac{\partial J^{(k)}}{\partial p_i} \right),
\normalsize
\end{equation*}
where the superscript $(k)$ denotes the $k$-th iteration\footnote{We write $\frac{\partial J(\bp^{(k)})}{\partial p_i}$ as $\frac{\partial J^{(k)}}{\partial p_i}$ and $\frac{\partial J(\bp^*)}{\partial p_i}$ as $\frac{\partial J^*}{\partial p_i}$ for simplicity.}, $s^{(k)}$ is the step size, and $\proj_{\Omega_i}(\cdot)$ is the projection operator. Since $\Omega_i$ are independent, we design the distributed projected gradient descent (DPGD) algorithm for coverage planning in Alg.\ref{alg:1}.

\begin{algorithm}
\KwInit: $p_i^{(0)} \gets 0, \exit_i \gets$ false, $\cnt_i \gets 0$\;
$k \gets 0$ ; \tcp{global clock}
\For{$k=1,2,\dots$}{
    listen($\exit_i, p_i^{(k)}$) ; \tcp{Execute when receiving query signals}
    \If{$\exit_i = \textnormal{false}$}{
        $\{ p_{-i}^{(k)}, \alpha_{-i}, \psi_{-i}^m \} \gets$ query() ; \tcp{query neighbors' info}
        
        Identify $C_i^{(k)}$ with $(\alpha_j, \psi_j^m, p_j^{(k)})$, $j \in \mN_i$ \;
        Compute $\frac{\partial J^{(k)}}{\partial p_i}$ with \eqref{eq:dJ_dpi} \;
        $p_i^{\proj} \gets \proj_{\Omega_i} \left( p_i^{(k)} - s^{(k)} \frac{\partial J^{(k)}}{\partial p_i} \right)$ \;
        \eIf{$\norm{\partial J^{(k)} / \partial p_i}_2 < \epsilon_i$ \KwOr $\norm{p_i^{\proj} - p_i^{(k)}}_2 < \epsilon_i$}{
            $\cnt_i \gets \cnt_i + 1$ \;
            $p_i^{(k+1)} \gets p_i^{(k)}$ \;
        }{
            $\cnt_i \gets 0$ \;
            $p_i^{(k+1)} \gets p_i^{\proj}$ \;
        }
        
        \If{$\cnt_i > \mathrm{max\_cnt}$ \KwOr $k > \mathrm{max\_k}$}{
            $\exit_i \gets$ true \;
            $p_i^\dd \gets p_i^{(k+1)}$ \;
        }
    }
    $k \gets k + 1$ \;
}
\KwFna{
    \eIf{$\exit_i = \textnormal{true}$}{
        broadcast $\{ p_i^\dd, \alpha_i, \psi^\mm_i \}$\;
    }{
        broadcast $\{ p_i^{(k)}, \alpha_i, \psi^\mm_i \}$ \;
    }
} 
\caption{DPGD algorithm for satellite $i$.}
\label{alg:1}
\end{algorithm}

\begin{remark}
In Alg.\ref{alg:1}, satellites perform computations at each global clock. The global clock only needs to be set once before the algorithm runs. Once the clock is set, each satellite communicates only with its adjacent neighbors to compute the new configuration. 
\end{remark}

\subsection{Convergence of DPGD Algorithm}
The following proposition guarantees the convergence of DPGD Alg.~\ref{alg:1} under mild conditions. 

\begin{proposition}
\label{prop:convergence}
The DPGD Alg.~\ref{alg:1} converges if all satellites adopt the same square-summable step size sequence $\{ s^{(k)}\}$, i.e., $\sum_{k=0}^\infty s^{(k)} = \infty$ and $\sum_{k=0}^\infty (s^{(k)})^2 < \infty$. Besides, the algorithm converges to the point $\bp^* = \{ p^*_1, \dots, p^*_n\}$ where either $p^*_i \in \bd \Omega_i$ or $\frac{\partial J^*}{\partial p_i} = 0$, $i \in \mN$.
\end{proposition}

\begin{proof}
From \eqref{eq:dJ_dpi} we observe that $\frac{\partial J}{\partial p_i}$ is bounded by some constant $L > 0$. Let $z \in \bOmega$ and $\xi \in \R$ and define $g: \R\to \R$ by $g(\xi) = J(\bp + \xi z)$. Using Tayler' theorem, we have
\begin{equation*}
\small
\begin{split}
    & J(\bp+z) - J(\bp) = g(1) - g(0) = \int_0^1 g'(\xi) \dd \xi \\ 
    = \ & \int_0^1 \frac{\partial J(\bp + \xi z)}{\partial \bp} z \dd \xi \\ 
    \leq \ & \int_0^1 \frac{J(\bp)}{\partial \bp} z \dd \xi + \abs{ \int_0^1 \left( \frac{\partial J(\bp + \xi z)}{\partial \bp} - \frac{J(\bp)}{\partial \bp} \right) z \dd \xi} \\ 
    \leq \ & \frac{J(\bp)}{\partial \bp} z + \int_0^1 \norm{z}_2 \norm{\frac{\partial J(\bp + \xi z)}{\partial \bp} - \frac{J(\bp)}{\partial \bp}}_2 \dd\xi \\ 
    \leq \ & \frac{J(\bp)}{\partial \bp} z + \frac{L}{2} \norm{z}^2_2.
\end{split}
\normalsize
\end{equation*}

For satellite $i$, given $p_i^{(k)}$, $p_i^{(k+1)}$ can be computed by $\frac{\partial J^{(k)}}{\partial p_i}$ and the projection. We denote $\Delta p_i^{(k)} = (p_i^{(k)} - p_i^{(k+1)}) / s^{(t)}$, which represents the actual negative descent direction. Since $\Omega_i$ represents a closed arc, $\Delta p_i^{(k)}$ always form an acute angle with $\frac{\partial J^{(k)}}{\partial p_i}$ when $\Delta p_i^{(k)} \neq 0$. 
Therefore, $\frac{\partial J^{(k)}}{\partial p_i} \Delta p_i^{(k)} \geq 0 $ always holds, and the equality is achieved when $\Delta p_i^{(k)} = 0$ or $\frac{\partial J^{(k)}}{\partial p_i} = 0$. Besides, we can also bound $\frac{\partial J^{(k)}}{\partial p_i} \Delta p_i^{(k)}$ with two positive numbers $0<m<M$ such that
\begin{equation*}
    m \norm{\Delta p_i^{(k)}}_2^2 \leq \frac{\partial J^{(k)}}{\partial p_i} \Delta p_i^{(k)} \leq M \norm{\Delta p_i^{(k)}}_2^2.    
\end{equation*}
Let $z_i = -s^{(k)} \Delta p_i^{(k)}$. Followed by the inequality, we have 
\begin{equation}
\label{eq:diff_J}
    J^{(k+1)} - J^{(k)} \leq \sum_{i=1}^n \left( -m s^{(k)} + \frac{L}{2} (s^{(k)})^2 \right) \norm{\Delta p_i^{(k)}}_2^2.
\end{equation}

For satellite $i$, when $\Delta p_i^{(k)} = 0$, there are two possibilities. If $\frac{\partial J^{(k)}}{\partial p_i} \neq 0$, then $p_i^{(k)} \in \bd \Omega_i$. The algorithm stops reducing $J$ and changing $p_i^{(k)}$. If $\frac{\partial J^{(k)}}{\partial p_i} = 0$, then $p_i^{(k)}$ is a stationary point. 
When $\Delta p_i^{(k)} \neq 0$, since $\{ s^{(k)} \}$ is a decreasing sequence, there exists $k > 0$ such that $\left( -m s^{(k)} + \frac{L}{2} (s^{(k)})^2 \right) > 0$ when $k > K$. So the algorithm constructs a decreasing sequence $\{ J^{(k)} \}$ and $J^{(k)}$ either converges to some finite value or $-\infty$. However, Prop.~\ref{prop:ne_exist} has shown that $J$ is bounded below. So $\lim_{k \to \infty} J^{(k)} > -\infty$. We sum all inequalities \eqref{eq:diff_J} when $k > K$ for all satellites and obtain 
\begin{equation*}
\begin{split}
    \sum_{i=1}^n \sum_{k=K}^\infty  m s^{(k)} \norm{\Delta p_i^{(k)}}_2^2 - \sum_{i=1}^n \sum_{k=K}^\infty \frac{L}{2} (s^{(k)})^2 \norm{\Delta p_i^{(k)}}_2^2 \\ 
    \leq J^{(K)} - \lim_{k\to\infty} J^{(k)}.
\end{split}
\end{equation*}
The right hand side is finite but the first summation in the left hand side can diverge because $\{s^{(k)}\}$ is square-summable. Therefore, we must have $\lim_{k\to\infty} \norm{\Delta p_i^{(k)}}_2 = 0$ for all $i \in \mN$. This shows that $J^{(k)}$ converges to some point $\bp^* = \{ p^*_1, \dots, p^*_n\}$ where either $p^*_i \in \bd \Omega_i$ or $\frac{\partial J^*}{\partial p_i} = 0$.
\end{proof}

\subsection{Stationary Point Analysis}
When the convergence is guaranteed, the DPGD Alg.~\ref{alg:1} generates some limiting point $\bp^*$ of \eqref{eq:Qp}. To verify whether it is a local NE of the coverage game $\mG$, we need to check whether $\bp^*$ is a real local minimum of the potential function $J$. 
Prop.~\ref{prop:convergence} has shown that some components of $\bp^*$ may belong to $\bd \bOmega$ while others have a vanishing gradient. In the following, we will first show a special case where $\bp^* \in \interior \bOmega$ (i.e., all satellites have a vanishing gradient) generates a local minimum. Then, we show that the general case, where some components of $\bp^*$ belong to $\bd \bOmega$, also outputs a local minimum.

When $\bp^* \in \interior \bOmega$, from \eqref{eq:dJ_dpi} we have
\begin{equation}
    \label{eq:stationary}
    \int_0^{T_s} \int_{C_i^+} \rho \dd \theta \dd \tau = \int_0^{T_s} \int_{C_i^-} \rho \dd \theta \dd \tau, \quad \forall i \in \mN.
\end{equation}
We also obtain for every $i \in \mN$
\begin{equation*}
    \frac{\partial^2 J}{\partial p_i^2} = \frac{1}{T_s} \int_0^{T_s} \int_{C_i} \left( \frac{\partial \psi_i}{\partial p_i} \right)^\tp \frac{\partial \psi_i}{\partial p_i} + \rho \frac{\partial^2 \psi}{\partial p_i^2} \dd \theta \dd \tau. 
\end{equation*}
By computing $\frac{\partial^2 \psi_i}{\partial p_i^2}$ from \eqref{eq:dpsi_dpi} and by referring to the stationary condition \eqref{eq:stationary}, we have 
\begin{equation*}
    \frac{1}{T_s} \int_0^{T_s} \int_{C_i} \rho^* \frac{\partial^2 \psi_i^*}{\partial p_i^2} \dd\theta \dd\tau = 0, \quad \forall i \in \mN,
\end{equation*}
where $\rho^*:= \rho(\bp^*, \theta, \tau)$ and $\psi_i^*:= \psi_i(p^*_i, \theta, \tau)$.
Therefore,
\begin{equation*}
    \small
    \frac{\partial^2 J^*}{\partial p_i^2} = \frac{2 \alpha_i k_i^2}{[(p_{ix}^*+r_s)^2 + p_{iy}^{*2}]^2} \begin{bmatrix} -p_{iy}^* \\ p_{ix}^* + r_s \end{bmatrix} \begin{bmatrix} -p_{iy}^* & p_{ix}^* + r_s \end{bmatrix}.
\end{equation*}
\normalsize
We also note that $\frac{\partial}{\partial p_j} \frac{\partial \psi_i}{\partial p_i} = 0$ for $j \neq i$. Hence 
\begin{equation*}
    \frac{\partial^2 J}{\partial p_i \partial p_j} = \frac{1}{T_s} \int_0^{T_s} \int_{\omega \tau}^{\omega \tau + 2\pi} \left( \frac{\partial \psi_j}{\partial p_j} \right)^\tp \frac{\partial \psi_i}{\partial p_i} \dd \theta \dd \tau.
\end{equation*}
From \eqref{eq:psi_i}, we see that $\frac{\partial^2 J}{\partial p_i \partial p_j} \neq 0$ if and only if satellite $i$ and satellite $j$ have an overlapped coverage region. Since only adjacent neighbors are considered, for satellite $i$, we have
\begin{equation*}
    \begin{split}
        &\frac{\partial^2 J}{\partial p_i \partial p_{i+1}} = \frac{1}{T_s} \int_0^{T_s} \int_{C_i \cap C_{i+1}} \left( \frac{\partial \psi_{i+1}}{\partial p_{i+1}} \right)^\tp \frac{\partial \psi_i}{\partial p_i} \dd \theta \dd \tau \\
        &= \left(\frac{\partial \psi_{i+1}}{\partial p_{i+1}} \right)^\tp \frac{\partial \psi_{i}}{\partial p_{i}} (\phi_i - \phi_{i+1} + \alpha_i + \alpha_{i+1}).
    \end{split}
\end{equation*}
Likewise, 
\begin{equation*}
\small
    \frac{\partial^2 J}{\partial p_i \partial p_{i-1}} = \left(\frac{\partial \psi_{i-1}}{\partial p_{i-1}} \right)^\tp \frac{\partial \psi_{i}}{\partial p_{i}} (\phi_{i-1} - \phi_{i} + \alpha_i + \alpha_{i-1}).
\end{equation*}\normalsize

A configuration involves multiple satellites and hence has different coverage scenarios. We divide all coverage scenarios into two categories and show that our DPGD algorithm generates the local NE for each of them.

\subsubsection{Full-Overlap Coverage Scenario}
In this scenario, all satellites share overlapped coverage regions with its neighbors. i.e., for satellite $i$, $C_i$ overlaps with both $C_{i+1}$ and $C_{i-1}$. 
Since only neighbors are coupled, the Hessian $\frac{\partial^2 J}{\partial \bp^2}$ has a banded structure. Then, we arrive at the following proposition to characterize the limiting point $\bp^*$. 

\begin{proposition}[Full-overlap]
\label{prop:full-overlap}
Let $\bp^* \in \R^{2n}$ be the limiting point generated by the DPGD Alg.~\ref{alg:1} and assume that $\bp^* \in \interior \bOmega$. Then, $\bp^*$ is a local minimum of $J$ defined in \eqref{eq:cost} if 
\begin{equation*}
    \abs{\phi_i^0 - \phi_{i+1}^0 + \Delta \phi_i - \Delta \phi_{i+1} + \alpha_i + \alpha_{i+1}} \leq \sqrt{\alpha_i \alpha_{i+1}}
\end{equation*}
for all $i = 1,\dots, n-1$. For $i=n$, the term $\phi^0_n - \phi^0_1$ is changed to $2\pi - \phi^0_n - \phi^0_1$ due to periodicity.
\end{proposition}

\begin{proof}
We first introduce the following lemma.
\begin{lemma}
\label{lemma:2}
Let $x, y \in \R^n$ and $\beta \in \R$. The matrix 
\begin{equation*}
     A = \begin{bmatrix} xx^\tp & -\beta xy^\tp \\ -\beta yx^\tp & yy^\tp \end{bmatrix}
\end{equation*}
is positive semidefinite if $\abs{\beta} \leq 1$.
\end{lemma}
The proof of Lemma~\ref{lemma:2} is in Appendix \ref{app:2}.

Let $x = \{ x_1, \dots, x_n \} \in \R^{2n}$ be an arbitrary vector with $x_i \in \R^2$ as the $i$-th component. Using the banded structure of the Hessian, we have 
\begin{equation*}
    \small
    \begin{split}
        x ^\tp \frac{\partial^2 J}{\partial \bp^2} x &= \sum_{i=1}^n x_i^\tp \frac{\partial^2 J}{\partial p_i^2} x_i + \sum_{i=1}^n  x_i^\tp \frac{\partial^2 J}{\partial p_i \partial p_{i+1}} x_{i+1} \\
        &\quad + x_{i+1}^\tp \frac{\partial^2 J}{\partial p_{i+1} \partial p_i} x_{i} \\
        &= \sum_{i=1}^n \begin{bmatrix} x_i^\tp & x_{i+1}^\tp \end{bmatrix}
        \begin{bmatrix}
            \frac{1}{2} \frac{\partial^2 J}{\partial p_i^2} & \frac{\partial^2 J}{\partial p_i \partial p_{i+1}} \\ \frac{\partial^2 J}{\partial p_{i+1} \partial p_i} & \frac{1}{2} \frac{\partial^2 J}{\partial p_{i+1}^2}
        \end{bmatrix}
        \begin{bmatrix} x_i \\ x_{i+1} \end{bmatrix} \\
        &:= \sum_{i=1}^n \begin{bmatrix} x_i^\tp & x_{i+1}^\tp \end{bmatrix} H_{i,i+1} \begin{bmatrix} x_i \\ x_{i+1} \end{bmatrix}.
    \end{split}
\end{equation*}
\normalsize
Here, $x_{n+1}$ and $p_{n+1}$ refer to $x_1$ and $p_1$ respectively. We denote $v_i = k_i \frac{\partial \Delta \phi_i}{\partial p_i}$ and $\Delta_{i,i+1} = \phi_i - \phi_{i+1}+\alpha_i + \alpha_{i+1}$ for $i \in \mN$. An exception is $i=n$. $\Delta_{n,n+1}$ refers to $\Delta_{n,1}$ and should be changed to $\Delta_{n,1} = 2\pi - \phi_n^0 + \phi_1^0 - \Delta \phi_n + \Delta \phi_1 + \alpha_1 + \alpha_n$ due to periodicity.

From the previous analysis, we have $\frac{\partial^2 J}{\partial p_i^2} = 2\alpha_i v_i v_i^\tp$ and $\frac{\partial^2 J}{\partial p_i \partial p_{i+1}} = - v_i v_{i+1}^\tp \Delta_{i,i+1}$. The matrix $H_{i,i+1}$ becomes 
\begin{equation*}
    H_{i,i+1} = \begin{bmatrix} \alpha_i v_i v_i^\tp & -\Delta_{i,i+1} v_i v_{i+1}^\tp \\ -\Delta_{i,i+1} v_{i+1} v_i^\tp & \alpha_{i+1} v_{i+1} v_{i+1}^\tp
    \end{bmatrix}.
\end{equation*}
Using Lemma~\ref{lemma:2}, we obtain that $H_{i,i+1} \succeq 0$ if $\abs{\Delta_{i,i+1}} \leq \sqrt{\alpha_i, \alpha_{i+1}}$. If all matrices $H_{i,i+1} \succeq 0$ for $i \in \mN$, it is clear that the Hessian $\frac{\partial^2 J}{\partial \bp}$ is positive semidefinite. Therefore, $\bp^*$ is a local minimum of $J$. 
\end{proof}

The condition in Prop.~\ref{prop:full-overlap} indicates that two adjacent satellites should not get too close to each other. More specifically, for any two adjacent satellites sharing an overlapped coverage region, neither of the satellite's coverage regions should contain the center point of the other satellite's coverage region. Otherwise, the stationary point may not be optimal. This condition can be easily satisfied through design in practice.

\subsubsection{Non-overlap Coverage Scenario}
When satellites are far away from each other and do not overlap in coverage regions, the global coverage intensity $\rho$ in $C_i$ is simply the local coverage intensity $\psi_i$ for $i \in \mN$. The following characterizes this special coverage scenario.

\begin{proposition}[Non-overlap]
\label{prop:non-overlap}
Let $\bp^* \in \R^{2n}$ be the limiting point generated by the DPGD Alg.~\ref{alg:1} and assume that $\bp^* \in \interior \bOmega$. Then, $\bp^*$ is a local minimum of $J$ in \eqref{eq:cost} if all the satellites do not have any overlap in coverage regions.
\end{proposition}

\begin{proof}
Since there is no overlap, the neighbors will not affect the coverage of satellite $i$. The second-order derivative $\frac{\partial^2 J}{\partial p_i \partial p_j} = 0$ for $j \neq i$. The Hessian $\frac{\partial^2 J}{\partial \bp^2}$ becomes a diagonal block matrix with $i$-th block as $\frac{\partial^2 J}{\partial p_i^2}$, which is a rank 1 matrix with positive coefficient and hence positive semidefinite. Therefore, $\frac{\partial^2 J}{\partial \bp^2}$ is also positive semidefinite, which shows that $\bp^*$ is a local minimum of $J$. 
\end{proof}

Proposition \ref{prop:non-overlap} indicates that there is no need to search for the new coverage configuration when satellites do not overlap in coverage regions. It is because the new configuration produces the same coverage performance as the current one. Steering the satellites away from the current non-overlapping configuration does not improve $J$. Instead, it will increase the fuel consumption. Therefore, the optimal strategy for non-overlapping scenarios is to keep them unchanged. 

\begin{remark}
When $\bp^* \in \interior \bOmega$, all satellites either have overlapped coverage regions with neighbors or have no overlaps at all. It means that satellites do not have one-sided overlapped coverage regions with neighbors. It can be seen by the stationary condition \eqref{eq:stationary}. However, one-sided coverage scenario is possible when $\bp^* \not\in \interior \bOmega$. Some satellites cannot go further if they reach the boundary of their feasible sets. We show in Sec.~\ref{sec:coverage_planning.general} that the DPGD Alg.~\ref{alg:1} can generate the local minimum of $J$ in more general coverage scenarios. 
\end{remark}

\subsection{General Coverage Scenarios} \label{sec:coverage_planning.general}
When the DPGD Alg.~\ref{alg:1} generates some limiting point $\bp^* \not\in \interior \bOmega$, we split the satellites into two groups based on whether $p^*_i \in \bd \Omega_i$ or not. Let $\mB = \{ i \in \mN: p^*_i \in \bd \Omega_i, \frac{\partial J^*}{\partial p_i} \neq 0 \}$ and $\overline{\mB} = \mN \backslash \mB$. 

For satellite $i \in \mB$, the stationary condition \eqref{eq:stationary} cannot be satisfied because $\norm{p^*_i}_2 = p^\mm$. However, we have the following proposition to characterize the satellites in $\mB$.

\begin{proposition}
\label{prop:local_opt}
Let $\bp^* \in \R^{2n}$ be the limiting point generated by the DPGD Alg.~\ref{alg:1}. Assume that the set $\mB \neq \varnothing$. Let $\delta_i(p_i) \subset \Omega_i$ be a small neighbor set containing $p_i$. Then, for satellite $i \in \mB$, we have
\begin{equation*}
    u_i(p^*_i, p^*_{-i}) \leq u_i(p_i', p^*_{-i}) \quad \forall p_i' \in \delta_i(p_i^*).  
\end{equation*}
\end{proposition}

\begin{proof}
From Alg.~\ref{alg:1}, for satellite $i \in \mB$, the limiting point $p^*_i$ has the property such that $p^*_i = \proj_{\Omega_i} (p^*_i - s \frac{\partial J^*}{\partial p_i})$ for $s > 0$. 
From Fig.~\ref{fig:proof_fig_1} we can see that $\Omega_i$ represents a closed arc. The projection always projects $p_i^*$ back to the same point. 
Also from \eqref{eq:dpsi_dpi}, we observe that $\frac{\partial \psi_i}{\partial p_i}$ is never parallel to the radial direction of the arc. So is $\frac{\partial J^*}{\partial p_i}$. Since $\frac{\partial J^*}{\partial p_i} \neq 0$, the negative gradient must lie in the outward space specified by the radial direction (the gray region in Fig.~\ref{fig:proof_fig_1} if $p_i^*$ is the left boundary point of $\Omega_i$).
Let $p_i' \in \delta_i(p_i^*)$. Then, $p_i' - p_i^*$ is the tangential direction of $\Omega_i$ at $p_i^*$ if $p_i'$ is close enough to $p_i^*$. Therefore, we have $\frac{\partial J^*}{\partial p_i} (p_i' - p_i^*) \geq 0$. Using Taylor's Theorem, we have 
\begin{equation*}
\begin{split}
    u_i(p_i', p^*_{-i}) - u_i(p^*_i, p^*_{-i}) &= J(\{p_i', p^*_{-i}\}) - J(\{p^*_i, p^*_{-i}\}) \\ 
    \ &= \sum_{i \in \mB} \frac{\partial J^*}{\partial p_i} (p_i' - p_i^*) \geq 0,
\end{split}
\end{equation*}
which completes the proof.
\end{proof}

\begin{figure}
    \centering
    \includegraphics[scale=0.43]{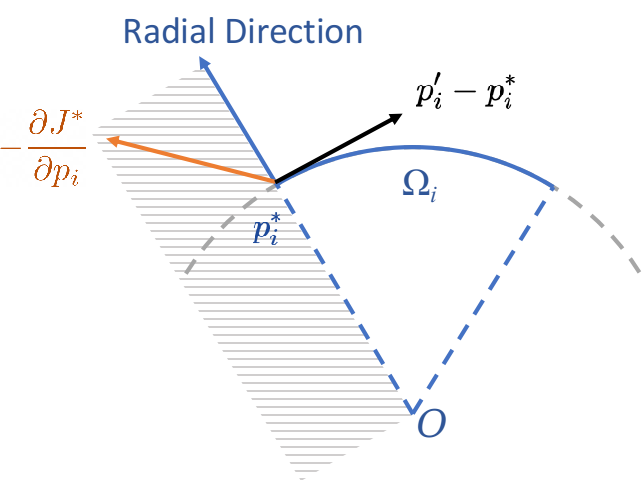}
    \caption{Sketch of $\Omega_i$ and a nonzero negative gradient $\frac{\partial J^*}{\partial p_i}$ at the left boundary point $p_i^*$ of $\Omega_i$. $p_i'-p_i^*$ forms a tangential vector if $p_i'$ is close enough to $p_i^*$.}
    \label{fig:proof_fig_1}
\end{figure}

Prop.~\ref{prop:local_opt} indicates that $\bp^*$ is already the local NE for satellite $i \in \mB$. For satellite $j \in \overline{\mB}$, since $p^*_j \in \interior \Omega_j$, we can use perturbation methods to check whether the rest $p^*_j$, $j \in \overline{\mB}$, constitute a local minimum of $J$. 
Let $x = \{ x_1, \dots, x_n \} \in \R^{2n}$ be an arbitrary vector with $x_i \in \R^2$ as the $i$-th component. We set $x_i = 0$ for $i \in \mB$ because the satellite $i \in \mB$ has already found a local NE and it has no incentive to deviate from the current solution. So there is no need to perturb satellite $i \in \mB$. Using the banded structure of the Hessian, we have
\begin{equation*}
\begin{split}
     & x^\tp \frac{\partial^2 J}{\partial \bp^2} x = \sum_{j \in \mathring{\mB}} x_j^\tp \frac{\partial^2 J}{\partial p_j^2} x_j 
     \\ & + \sum_{j, j+1 \in \overline{\mB}} \begin{bmatrix} x_j^\tp & x_{j+1}^\tp \end{bmatrix}
        \begin{bmatrix}
            \frac{1}{2} \frac{\partial^2 J}{\partial p_j^2} & \frac{\partial^2 J}{\partial p_j \partial p_{j+1}} \\ \frac{\partial^2 J}{\partial p_{j+1} \partial p_j} & \frac{1}{2} \frac{\partial^2 J}{\partial p_{j+1}^2}
        \end{bmatrix}
        \begin{bmatrix} x_j \\ x_{j+1} \end{bmatrix}.
\end{split}
\end{equation*}
Here, $\mathring{\mB} \subset \overline{\mB}$ represents the set of satellites adjacent to the satellites in $\mB$. The second summation requires both satellites $j$ and $j+1$ in $\overline{\mB}$. 
Therefore, if the condition in Prop.~\ref{prop:full-overlap} satisfies for satellite $j \in \overline{\mB}$, the Hessian $\frac{\partial J}{\partial \bp}$ is positive semidefinite, which further proves that $\bp^*$ is a local minimum of $J$. Hence $\bp^*$ constitutes a local NE for all satellites. 

\textcolor{black}{To summarize, the DPGD Alg.~\ref{alg:1} can generate a local NE $\bp^*$ of the coverage game $\mG$ under the following conditions:} for satellite $j \in \overline{\mB}$, i.e., $p^*_j \in \interior \Omega_j$, the condition in Prop.~\ref{prop:full-overlap} needs to be met; for satellite $i \in \mB$, i.e., $p^*_i \in \bd \Omega_i$, there is no extra condition.

\section{Distributed Coverage Control Synthesis} \label{sec:synthesis}
In this section, we propose the multi-waypoint model predictive control (mwMPC) for distributed and resilient satellite constellation control. We integrate the coverage game and the controller to develop the overall distributed framework for resilient satellite constellation coverage planning and control.

\subsection{Multi-Waypoint Model Predictive Control}
After obtaining the target configuration $\bp^\dd$ from the coverage game $\mG$, satellites control themselves autonomously to the target position by solving the constrained-LQR problem \eqref{eq:Qci}. Due to the convex input constraint, the analytic solution is hard to obtain. We reformulate \eqref{eq:Qci} into its discrete counterpart and solve it efficiently. Let $\Delta t \in \R_+$ be the sampling period interval and $N = T_f / \Delta t$. We denote $p_{i,k} \in \R^2, v_{i,k}\in \R^2$, and $u_{i,k}\in \R^2$ as the position, velocity, and external controls of satellite $i$ at time step $k$, $k=0,1,\dots, N-1$. We write $q_{i,k} = \{ p_{i,k}, v_{i,k} \} \in \R^4$, $q^\dd_i = \{ p^\dd_i, 0_{2\times 1} \} \in \R^4$, and $(A_d, B_d)$ as the corresponding discrete system dynamics of \eqref{eq:control}. We also use soft constraints to approximate the terminal constraints in \eqref{eq:Qci}. Let $\tilde{Q}_i \in \mathbb{S}^{4\times 4}$ be the augmented positive definite penalty matrix. The discrete counterpart can be written as
\begin{equation}
\small
    \label{eq:tilde_Qci}
    \tag{$\tilde{\mQ}_{ci}$}
    \begin{split}
        \min_{u_i} \quad & \norm{q_{i,N} - q_i^\dd}_{\tilde{Q}_i}^2 + \sum_{k=0}^{N-1} \norm{q_{i,k}-q_i^\dd}_{\tilde{Q}_i}^2 + \norm{u_{i,k}}_{R_i}^2 \\
        \text{s.t.} \quad & q_{i,k+1} = A_d q_{i,k} + B_d u_{i,k}, \quad k = 0, \dots, N-1, \\ 
        & u_{i,k}^\tp u_{i,k} \leq (u_{i}^\mm)^2, \quad k = 0, \dots, N-1.
    \end{split}
\normalsize
\end{equation}

In practice, satellite $i$ may not reach the target position $p^\dd_i$ after the time $T_f$ due to the small thrust. However, new attacks/incidents can happen during the satellite maneuver, causing further degeneration in the coverage performance. To cope with these issues, we can set multiple waypoints along the satellite trajectory. At each waypoint, attack detection is enabled so that the satellite can readjust its control to deal with the threat in time. 
More specifically, satellite $i$ sets $W_i$ waypoints $\{ \tilde{p}_i^{(m)} \}_{m=1}^{W_i}$ with $\tilde{p}_i^{(W_i)} = p_i^\dd$. If new threats occurs during the maneuver, the satellite goes to the nearest waypoint and restarts the DPGD algorithm to adapt to the new security environment. The mwMPC algorithm is designed in Alg.\ref{alg:2}.

\begin{algorithm}
\KwInit $p^\dd_i \gets$ DPGD Alg.~\ref{alg:1}, $\mathrm{atk}_i \gets $ false \; 
$\{ \tilde{p}_i^{(m)} \}_{m=1}^{W_i} \gets$ set $W_i$ waypoints \;
$m \gets 1$ ; \tcp{waypoint index} 
\For{$m = 1,\dots, W_i$}{
    $\tilde{q}_i^{(m)} \gets \{\tilde{p}_i^{(m)}, 0_{2\times 1}\}$ \;
    \While{$\norm{q_{i,N} - \tilde{q}_i^{(m)}}_2 \geq \epsilon$}{
        solve \eqref{eq:tilde_Qci} with $q^\dd_i = \tilde{q}_i^{(m)}$ \;
        $q_{i,0} \gets q_{i,1}$\;
    }
    \If{ \textnormal{attack\_detection()} }{
        $\mathrm{atk}_i \gets$ true \;
        broadcast $\mathrm{atk}_i$ \;
        break \;
    }
}

\KwFnb{
    \uIf{satellite is attacked}{
        \KwReturn true \;
    }
    \uElseIf{$\exists j \in \mN, j\neq i$, $\mathrm{atk}_j = $ true }{
    \KwReturn true \;
    }
    \Else{
    \KwReturn false \;
    }
}
\caption{mwMPC algorithm for satellite $i$.}
\label{alg:2}    
\end{algorithm}

\subsection{Distributed Coverage Planner and Controller Synthesis}
We consolidate the coverage game $\mG$ and mwMPC controller into the following DPGD-mwMPC framework.

\begin{algorithm}
    $\bp^d \gets$ DPGD planning algorithm (Alg. \ref{alg:1}) \;
    \For{Satellite $i=1$ \KwTo $n$ (in parallel)}{
        Receive $p^\dd_i$ \;
        run mwMPC$(p^\dd_i)$ algorithm (Alg. \ref{alg:2}).
    }
    \If{new threat detected}{
        wait for all satellites go to nearest waypoint \;
        \KwGoto DPGD planning algorithm (line 1)\;
    }
\caption{DPGD-mwMPC framework}
\label{alg:3}
\end{algorithm}

The DPGD-mwMPC framework enables a more flexible and resilient approach for space threats that affect coverage performance. Satellites can react to the new security environment in time at any waypoint. Besides, satellites can also assess the security level of the environment by choosing the number of waypoints. If the environment is secure enough, all satellites can simply set one waypoint during the maneuver. Otherwise, multiple waypoints can be set to monitor threats in real-time.

\section{Case Studies} \label{sec:case_study}
In this section, we demonstrate the resilience of our framework by experimenting with different types of space attacks. We consider a single-orbit satellite constellation with $n = 25$ homogeneous LEO satellites. Each satellite has the same coverage parameters $(\alpha, \psi^\mm)$ at the beginning. The constellation becomes heterogeneous when some satellites are attacked. We normalize the units here as some parameters are huge such as the earth radius. We define 1 distance unit (DU) as $10^6$m and 1 time unit (TU) as 100s. We set the orbital altitude $h=800$km ($=0.8$DU) and $FOV=48^\circ$, which are typical values of LEO satellites. The maximum thrust-to-mass ratio for each satellite is set as $u^\mm = 0.01$ DU/TU$^2$. The demand intensity $\mu(\theta)$ can be obtained by data in practice and we use a truncated multimodal normal distribution on $[0,2\pi)$ in the case study. We also normalize the coverage intensity by setting $\psi^\mm = 10$ for all satellites so that $\int_0^{2\pi} \mu(\theta) \dd\theta = \int_0^{2\pi} \rho(\bp, \theta, 0) \dd\theta$. The sampling period $\Delta t = 0.6$ TU and the control horizon $T_f = 18$ TU.

\subsection{Coverage under Cyber Attacks}
Cyber attacks such as jamming attacks can degenerate the coverage performance by reducing the FOV angle and the maximum coverage intensity $\psi^\mm_i$. Cyber attacks generally do not destroy the physical equipment in the satellite, and the attacked satellite may recover to some extent when the attack is over. Therefore, we consider the following attack-recovery plan for the satellite constellation, where $6$ satellites are attacked and recovered. The attacked satellite's coverage parameters are changed according to Tab.~\ref{tab:cyberattack}. 
They partially recover using the coverage parameters in recovery plan I and get fully recovered using recovery plan II.

\begin{table}[h]
    \centering
    \begin{tabular}{c|c|c|c}
        \hline
        \multirow{2}{*}{sat.\#} & attack & recovery I & recovery II \\ 
         & ($FOV,\psi_i^m$) & ($FOV,\psi_i^m$) & ($FOV,\psi_i^m$) \\ \hline
        1  & $(44,8)$ & $(48,10)$ & $(48,10)$ \\ 
        2  & $(42,7)$ & $(47,9)$  & $(48,10)$ \\
        3  & $(42,7)$ & $(44,8)$  & $(48,10)$ \\ \hline
        14 & $(44,8)$ & $(47,9)$  & $(48,10)$ \\ 
        15 & $(42,6)$ & $(45,7)$  & $(48,10)$ \\ \hline
        12 & $(46,8)$ & $(48,10)$ & $(48,10)$ \\ \hline
    \end{tabular}
    \caption{Attack-recovery plan for space cyber attacks.}
    \label{tab:cyberattack}
\end{table}

We divide the entire process into three phases: initialization, attack, and recovery phases. In the initialization phase, all satellites find a local optimal configuration given the IC. The IC can be arbitrarily assigned or can be the terminal configuration from the previous attack process, and the IC may not be optimal under the current security environment. The attack and the recovery phases show the reactions of the satellite constellation after the cyber attack and recovery, respectively. The entire processes of the two recovery plans are shown in Fig.~\ref{fig:cyberattack}. The coverage cost is computed based on the relative position vector $\bp$ at each control step.
Three phases are distinguished by the jumps in the coverage cost, which shows the destructiveness of the cyber attack to the coverage performance. 
Our framework shows the resilience of satellite control in all three phases. All satellites can not only adapt to the given IC, but also mitigate the attack consequence and reach a new local optimal configuration. The zoomed plots show that all satellites seek new configurations to actively improve coverage performance.
In the recovery phase, the attack is over and the attacked satellites recover partial (and full) coverage capability. Then, all satellites readjust the configuration based on the recovered capabilities to provide better coverage performance. Fig.~\ref{fig:cyberattack.2} successfully demonstrates that the coverage performance converges to the pre-attack level after the full recovery. 

In both experiments, we make the attacked satellites recover from the attack after they form a new local optimal configuration. It is to show the convergence of our DPGD algorithm. 

\begin{remark}
From Fig.~\ref{fig:cyberattack}, we learn that in the coverage planning stage, all satellites do not move and first communicate with neighbors to compute the new configuration using the DPGD algorithm. This process converges after many iterations and generates the target configuration $\bp^\dd$. Then satellites move to $\bp^\dd$ by using controls. Therefore, the convergence of the DPGD algorithm is shown in Fig.~\ref{fig:cyberattack}. 
\end{remark}

\begin{figure}
    \centering
    \begin{subfigure}[t]{0.45\textwidth}
        \centering
        \includegraphics[height=4.5cm]{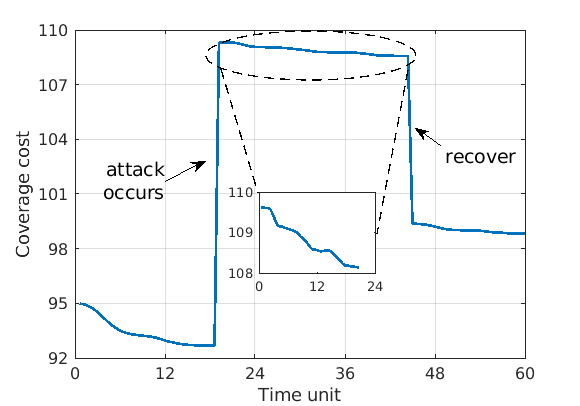}
        \caption{Partial recovery.}
        \label{fig:cyberattack.1}
    \end{subfigure}
    \hspace{5mm}
    \begin{subfigure}[t]{0.45\textwidth}
        \centering
        \includegraphics[height=4.5cm]{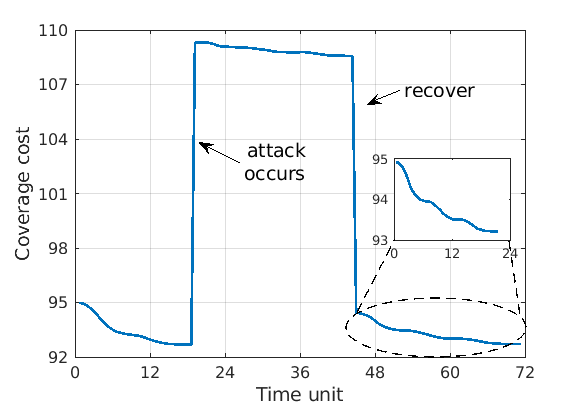}
        \caption{Full recovery.}
        \label{fig:cyberattack.2}
     \end{subfigure}
    \caption{Coverage cost evolution under cyber attacks.}
    \label{fig:cyberattack}
\end{figure}

\subsection{Cyber Attacks with Anchored Satellites}
In practice, some satellites are preferred not to actively move in a configuration. These satellites can be either the ones that have already run out of fuel, or the ones with more sophisticated payloads and need more fuel to adjust positions. These satellites reduce the feasible configurations when attacks occur, and the satellite constellation have more constraints to search for a new configuration to adapt to the threats. The discovered configuration may be less resilient for coverage control. We consider the same satellite constellation with four anchored satellites (satellite $\{6,12,18,24\}$) and use the same attack-recovery plan in Tab.~\ref{tab:cyberattack}. 

We compare the coverage cost evolution of anchored and non-anchored satellite constellations in Fig.~\ref{fig:anchor.1}. An immediate observation is the anchored constellation provides a higher coverage cost (poorer coverage performance) in each phase compared with the non-anchored constellation. It is because several satellites lose their mobility to jointly pursue a better configuration to adapt to attacks. 
We can observe from the attack phase in Fig.~\ref{fig:anchor.1} that the anchored constellation, although producing poorer coverage performance, can reach the new local optimal configuration faster than the non-anchored constellation, which means that anchored satellites result in less overall control effort. Indeed, Fig.~\ref{fig:anchor.2} compares the total control cost in three phases between anchored and non-anchored constellations. It shows the fundamental trade-off between the coverage performance (the resilience) and control cost. The more anchored satellites we have, the fewer controls we need to adapt to new attacks, but the less resilience we can obtain. 

On the other hand, the anchored constellation also improves time flexibility. If the coverage performance can be tolerated within a range, we can set some anchored satellites to react to attacks in a faster and fuel-saving way. \textcolor{black}{Since the response time is reduced, the anchored constellation is more flexible for the consecutive attacks.}
The anchored satellite constellation provides us a way to consider the trade-off between the resilience in coverage performance and the control cost under cyber attacks. If the coverage task has a higher priority, we need to dispatch all available satellites to actively mitigate the attack consequence. Otherwise, we can create several anchored satellites to reduce the control efforts. Whatever the situation we have, our framework can provably guarantee resilient coverage planning and control.

\begin{figure}
    \centering
    \begin{subfigure}[t]{0.45\textwidth}
        \centering
        \includegraphics[height=4.5cm]{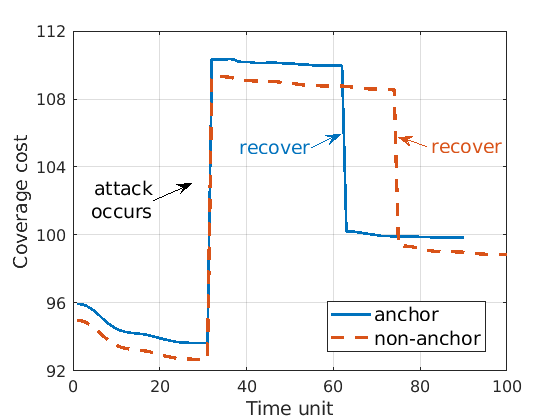}
        \caption{Coverage cost comparison.}
        \label{fig:anchor.1}
    \end{subfigure}
    \hspace{5mm}
    \begin{subfigure}[t]{0.45\textwidth}
        \centering
        \includegraphics[height=4.5cm]{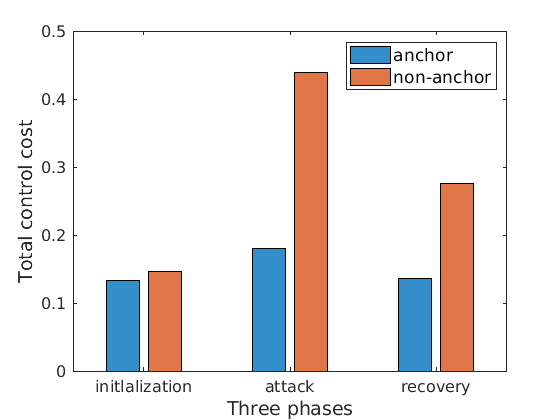}
        \caption{Control cost comparison.}
        \label{fig:anchor.2}
     \end{subfigure}
    \caption{Performance comparison for anchored and non-anchored satellite constellations.}
    \label{fig:anchor}
\end{figure}

\subsection{Coverage under Satellite Loss and Replenishment}
Many space threats such as debris and laser attacks can jeopardize the satellite structure or disable the satellite, resulting in satellite losses. The spare satellites have to be used to replenish the current configuration to keep providing normal coverage service. In this case, we consider the following loss-recovery plan for the satellite constellation. The index of disabled satellites and their neighbors are listed in Tab.~\ref{tab:neighbor_attack}. Two recovery plans replenish the configuration with two and three spares, respectively. The details of the spare's neighbors and coverage parameters are also included. Fig.~\ref{fig:neighbor_sketch} shows a sketch of how the satellites are attacked and replenished.

\begin{table}[h]
    \centering
    \begin{tabular}{c|c|c}
        \hline
        attack & recovery I & recovery II \\
        sat.\#-(nbr.\#) & sat.\#-(nbr.\#): ($FOV,\psi_i^m$) & sat.\#-(nbr.\#): ($FOV,\psi_i^m$) \\ \hline
        1-$(22,25)$  & \multirow{2}{*}{26-$(27,25)$: $(48,10)$} & 26-$(3,25)$: $(48,10)$ \\ 
        2-$(3,1)$    &     & 27-$(3,26)$: $(48,9)$ \\ \hline
        12-$(13,11)$ & \multirow{2}{*}{27-$(14,11)$: $(44,8)$} & \multirow{2}{*}{27-$(14,11)$: $(44,8)$} \\ 
        13-$(14,12)$ &     & \\ \hline
        22-$(23,21)$ & N/A & N/A \\ \hline
    \end{tabular}
    \caption{Loss-recovery plan for physical attacks.}
    \label{tab:neighbor_attack}
\end{table}

\begin{figure}[h]
    \centering
    \includegraphics[height=3.1cm]{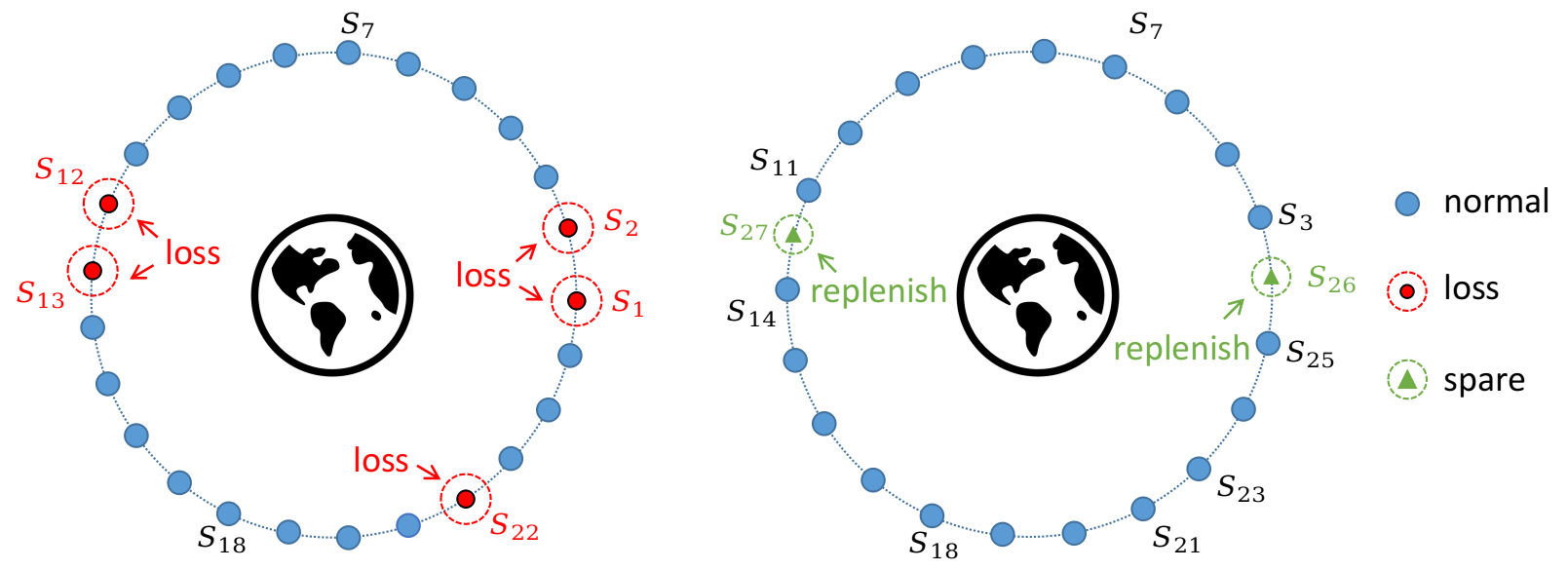}
    \caption{Sketch of loss-recovery plan.}
    \label{fig:neighbor_sketch}
\end{figure}

To demonstrate the team and individual attacks, we divide the attacked satellites into three groups. The first two groups contain two satellites, respectively, representing the team attack. Fig.~\ref{fig:neighbor.1} shows the coverage performance in three phases according the loss-recovery plan I in Tab.~\ref{tab:neighbor_attack}. After the initialization phase, the coverage cost dramatically increases due to satellite losses. The rest satellites start to adapt to attack and find a new local optimal configuration. After reaching the optimal configuration to mitigate the attack, we recover the configuration with two spares. The replenished satellite constellation begins to move to a better configuration for better coverage performance. We notice that the adaptation time in Fig.~\ref{fig:neighbor.1} is much longer than the one in Fig.~\ref{fig:cyberattack}. It is because physical attacks cause more severe damage to the constellation and the satellites have to use longer time and more control effort to mitigate the damage. Although the attack is more extreme, our DPGD algorithm can still provide convergent solutions and has successfully driven all satellites to the local optimal configuration.

We also compare the recovery phase of two recovery plans in Tab.~\ref{tab:neighbor_attack} and plot the coverage cost evolution of the recovery phase in Fig.~\ref{fig:neighbor.2}. We observe that using more spares to cope with physical attacks can generate better coverage performances and faster adaptation. This result also corroborates the one key principle: more redundancy leads to more resiliency.   

\begin{figure}
    \centering
    \begin{subfigure}[t]{0.45\textwidth}
        \centering
        \includegraphics[height=4.5cm]{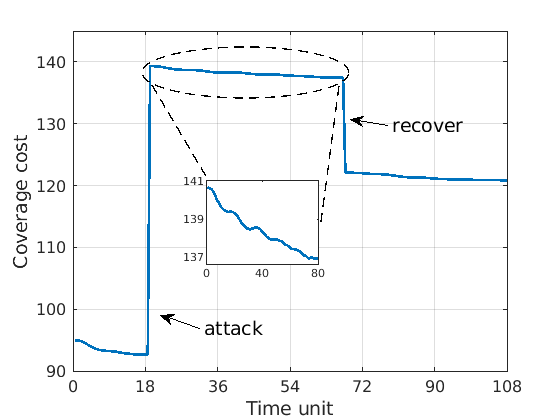}
        \caption{Recover 2 satellites.}
        \label{fig:neighbor.1}
    \end{subfigure}
    \hspace{5mm}
    \begin{subfigure}[t]{0.45\textwidth}
        \centering
        \includegraphics[height=4.5cm]{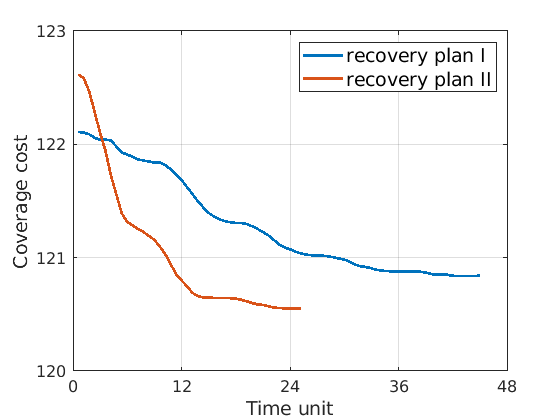}
        \caption{Coverage cost comparison.}
        \label{fig:neighbor.2}
     \end{subfigure}
    \caption{Coverage performance and comparison under physical attacks.}
    \label{fig:neighbor}
\end{figure}

\subsection{Comparison with Equal-Spacing Control Strategy}
We compare our DPGD-mwMPC framework with the equal-spacing control strategy, where all satellites aim to maintain an equal space with each other in the configuration. The equal-spacing control can be viewed as an averaging dynamical system. By using the ring structure of the satellite constellation, we can explicitly compute the final position given any IC. We experiment on the satellite constellation's reaction to the cyber attack in Tab.~\ref{tab:cyberattack} by using our framework and the equal-spacing control strategy. The results are shown in Fig.~\ref{fig:equalspace}. 

Compared with our framework, the equal-spacing control strategy provides a higher coverage cost (poorer coverage performance) to adapt to the attack, although it shows more time flexibility at the beginning. Thus, it is clear that an equal-spacing control strategy requires less control cost, shown in Fig.~\ref{fig:equalspace.2}.
However, one central weakness of the equal-spacing control strategy is that it provides zero resilience and robustness when the security environment changes. As shown in Fig.~\ref{fig:equalspace.1}, when the attacked satellites recover from the cyber attack, the equal-spacing strategy simply keeps the current configuration, showing no adaptation. For consecutive cyber attacks, the equal-spacing strategy does not respond either. For comparison, our framework has shown to be much more resilient and flexible for various space attacks. 

\begin{figure}
    \centering
    \begin{subfigure}[t]{0.45\textwidth}
        \centering
        \includegraphics[height=4.5cm]{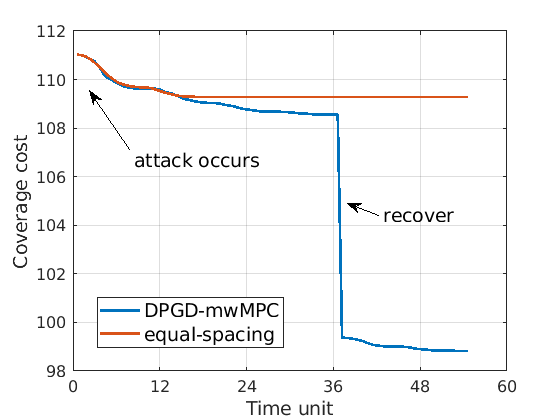}
        \caption{Coverage cost comparison.}
        \label{fig:equalspace.1}
    \end{subfigure}
    \hspace{5mm}
    \begin{subfigure}[t]{0.45\textwidth}
        \centering
        \includegraphics[height=4.5cm]{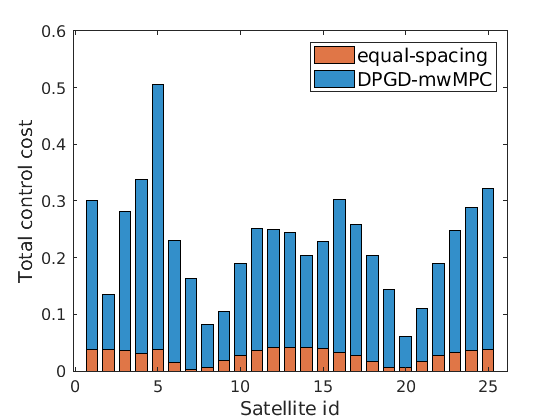}
        \caption{Control cost comparison.}
        \label{fig:equalspace.2}
     \end{subfigure}
    \caption{Comparisons of DPGD-mwMPC framework and equal-spacing control strategy under cyber attacks.}
    \label{fig:equalspace}
\end{figure}

\section{Conclusion} \label{sec:conclusion}
In this paper, we have investigated the satellite constellation coverage under different space security threats by establishing an integrative, distributed, and resilient planning-control framework. The proposed framework has not only captured the multi-objective of maximizing coverage performance and minimizing the total control effort for satellite constellation adaptation, but also improved the resilience of the satellite constellation coverage for adversarial and non-adversarial attacks. The proposed coverage game and the agent-based algorithm have shown effectiveness in searching for optimal coverage configuration in different coverage scenarios. The optimality of the coverage configuration and the convergence of the algorithm are also discussed in detail. The multi-waypoint MPC has achieved fuel-optimal control and provided more flexibility for the satellite constellation to deal with consecutive attacks. Case studies have demonstrated that our DPGD-mwMPC framework provides a solid resilience to cyber and physical attacks and outperforms the equal-spacing control strategy. For future work, we would consider the coordination and resilient control of multi-orbit satellite constellations. We would also investigate the impact of time-varying ground demand intensity on the satellite constellation's adaptation strategy.

\appendices
\section{Proof of Lemma \ref{lemma:1}}\label{app:1}
\begin{proof}
Using the fundamental theorem of calculus, we write $F(x) = \int_0^x f(\tau) \dd \tau + F(0)$ for some value $F(0)$. Then we have
\begin{equation*}
    G(x) = F(x+\delta) - F(x) = \int_{x}^{x+\delta} f(\tau) \dd\tau.
\end{equation*}
Note that 
\begin{equation*}
    \begin{split}
        G(x+T) &= F(x+\delta+T) - F(x+T) \\
        &= \int_{x+T}^{x+\delta+T} f(\tau) \dd \tau \overset{\tau=u-T}{=} \int_{x}^{x+\delta} f(u-T) \dd u \\
        &= \int_{x}^{x+\delta} f(u) \dd u = G(x),
    \end{split}
\end{equation*}
where we the use change of variable in the second line. So $G(x)$ is periodic with period $T$. We also have 
\begin{equation*}
    \begin{split}
        &\int_0^T G(x+\epsilon) \dd x \overset{u=x+\epsilon}{=} \int_{\epsilon}^{\epsilon+T} G(u) \dd u \\
        =& \int_0^T G(u) \dd u + \int_T^{T+\epsilon} G(u) \dd u - \int_0^\epsilon G(u) \dd u \\
        =& \int_0^T G(u) \dd u + \int_0^{\epsilon} G(z) \dd z - \int_0^\epsilon G(u) \dd u \\
        =& \int_0^T G(u) \dd u,
    \end{split}
\end{equation*}
where we use the change of variable $z=u-T$ in the third row. 
\end{proof}

\section{Proof of Lemma \ref{lemma:2}} \label{app:2}
\begin{proof}
Let $v \in \R^{2n}$ be an arbitrary vector. We split $v= [v_1 \ v_2]$ with $v_1, v_2 \in \R^n$. We want to show $v^\tp A v \geq 0$ for all $v \in \R^{2n}$. Indeed, we have 
\begin{equation*}
    \begin{split}
        v^\tp A v =& (v_1^\tp x)^2 + (v_2^\tp y)^2 - 2\beta (v_1^\tp x) (v_2^\tp y) \\
        :=& a^2 + b^2 -2\beta ab = f(a,b). 
    \end{split}
\end{equation*}
The result clearly holds when $x= 0$ or $y = 0$. When $x, y \neq 0$, we note that $a,b$ can be arbitrary values because $v_1$, $v_2$ are arbitrary vectors. To make $v^\tp A v \geq 0$, we need $f(a,b)$ be convex in both $(a,b)$ and $f_{\min} \geq 0$. By checking the Hessian of $f$, we can obtain $\abs{\beta} \leq 1$. 
\end{proof}


\ifCLASSOPTIONcaptionsoff
  \newpage
\fi

\bibliographystyle{IEEEtran}
\bibliography{satellite_resilient}

\end{document}